\documentclass[lettersize,journal]{IEEEtran}
\usepackage{amsmath,amsthm,graphicx}
\usepackage{comment}
\usepackage{algorithm,algorithmic}     
\usepackage{caption}
\usepackage{subcaption}
\usepackage{comment}
\usepackage{adjustbox}
\usepackage{amsfonts}
\usepackage[dvipsnames]{xcolor}
\graphicspath{ {./images/} }
\newtheorem{theorem}{Theorem}
\newtheorem{proposition}{Proposition}
\newtheorem{lemma}{Lemma}

\begin{document}

\title{Community detection in multiplex networks based on orthogonal nonnegative matrix tri-factorization}

\author{Meiby~Ortiz-Bouza,~\IEEEmembership{}
        Selin~Aviyente~\IEEEmembership{}
\thanks{ M. Ortiz-Bouza and S. Aviyente were with the Department
of Electrical and Computer Engineering, Michigan State University, East Lansing,
MI, 48824.\protect\\
E-mail: ortizbou@msu.edu, aviyente@egr.msu.edu
}
\thanks{
}}

\markboth{
}%
{Shell \MakeLowercase{\textit{et al.}}: Community detection in multiplex networks based on orthogonal nonnegative matrix tri-factorization}
\newcommand{\matr}[1]{\mathbf{#1}}

\IEEEtitleabstractindextext{%
    \begin{abstract}
         Networks are commonly used to model complex systems.  The different entities in the system are represented by nodes of the network and their interactions by edges. In most real life systems, the different entities may interact in different ways necessitating the use of multiplex networks where
         multiple links are used to model the interactions.
        One of the major tools for inferring network topology is community detection. Although there are numerous works on community detection in single-layer networks, existing community detection methods for multiplex networks mostly learn a common community structure across layers and do not take the heterogeneity across layers into account. In this paper, we introduce a new multiplex community detection method that  identifies communities that are common across layers as well as those that are unique to each layer. The proposed method, Multiplex Orthogonal Nonnegative Matrix Tri-Factorization, represents the adjacency matrix of each layer as the sum of two low-rank matrix factorizations  corresponding to the common and private communities, respectively. Unlike most of the existing methods which require the number of communities to be pre-determined, the proposed method also introduces a two stage method  to determine the number of common and private communities. The proposed algorithm is evaluated on synthetic and real multiplex networks, as well as for multiview clustering applications, and compared to state-of-the-art techniques.
    \end{abstract}
        
}

\maketitle

\IEEEdisplaynontitleabstractindextext
\IEEEpeerreviewmaketitle

\section{Introduction}\label{sec:introduction}
\IEEEPARstart{C}{omplex} networks are usually used to represent many real world systems, ranging from social to biological ones  \cite{barabasi2013network}, where the different agents and their relations are represented as the nodes and edges of the network, respectively. Traditional network models  employ simple graphs, where there is a single edge between any two nodes. Thus, these models cannot capture multiple modes of interaction that may exist between the nodes. Recently, multiplex networks that represent multiple modes of interaction have been proposed. A multiplex network is a  multilayer network where all layers share the same set of nodes with different topologies \cite{kivela_multilayer_2014}. Multiplex networks have been used to model a variety of complex systems, including living organisms, human societies, transportation systems and critical infrastructures \cite{smith2019using,aleta2017multilayer}.

Community detection is an important tool in network analysis, where communities are defined as groups of nodes that are more densely connected to each other than they are to the rest of the network. Most of the existing work on community detection \cite{fortunato2016community} focuses on single layer networks. Community detection methods for multiplex networks \cite{magnani2021community} can be grouped into three main classes. The first class of methods merges the layers in a multiplex network, using a flattening algorithm, then apply single-layer community
detection to the aggregated network \cite{berlingerio2011finding,chen2017multilayer,taylor2017super}. While these methods are computationally efficient, they can only identify communities that are common across all layers. Moreover, due to the flattening process, some spurious communities may emerge. The second class of methods applies community detection to each layer individually and then merges the results \cite{berlingerio2013abacus,tang2012community,dong2013clustering}. 
These methods include nodes in the same community only when they are part of the same community in at least one layer. 
Finally, the third class of methods operates directly on the multiplex network model \cite{kuncheva2015community,de2015identifying,mucha_community_2010,amelio2014community,zhu2014unified,roxana_pamfil_relating_2018}. 

Existing multiplex community detection approaches typically assume that the community structure is the same across layers and find the partition that best fits all layers. Thus, they do not differentiate between communities that are common across layers from those that are unique to each layer. This is particularly important for real world applications where the networks are heterogeneous, and the different layers correspond to different modes of interaction. For example, in social networks, a group of individuals may be well connected via friendships on Facebook; however, this group of individuals will likely not work at the same company. Thus, in a situation like this, a given  community will only be present in a subset of the layers, and different communities may be present in different subsets of layers. 

In this work, common communities are defined as communities that are observed in more than one layer, i.e., communities that are common across any subset of two or more layers, and private communities as communities that are unique to each layer. The problem of detecting  common and private communities is then formulated using a novel framework titled Multiplex Orthogonal Nonnegative Matrix Trifactorization (MX-ONMTF). In the proposed framework, each layer's adjacency matrix is represented as the sum of two low-rank matrix factorizations corresponding to the common and  private communities, respectively. The resulting optimization problem is solved using an iterative multiplicative update algorithm. The proposed approach also addresses the problem of determining the number of communities. Unlike most existing work, where the number of communities is determined through a greedy search, in this paper, a two-step approach is proposed. 
The proposed algorithm is first evaluated on synthetic benchmark multiplex networks with different numbers of layers, nodes, communities, noise levels, and inter-layer dependency probability. Next,  the proposed method is applied to real networks including social and biological networks. Finally, the algorithm is evaluated for multiview clustering task, where the communities across all layers are assumed to be the same.  

This paper extends our prior work \cite{ortiz2022orthogonal} where the community detection problem is formulated for two-layer multiplex networks. This work introduces significant new contributions. First, MX-ONMTF is generalized  for multiplex networks with $L$ layers where the common communities can be observed for a subset of layers. Second, a new two stage method is introduced to determine the number of common communities such that communities that are common across two or more layers can be identified. A detailed theoretical analysis of the convergence of the algorithm and recovery guarantees is also provided. Finally, this paper includes an extensive evaluation of the method on synthetic and real networks, and multiview clustering applications. 

The rest of the paper is organized as follows. Section \ref{sec:relworks} presents a summary of related works. Section \ref{sec:background} provides background on community detection, multiplex networks, and orthogonal nonnegative matrix tri-factorization. Sections \ref{sec:mxonmtf} and \ref{sec:convergence} present the proposed multiplex community detection algorithm and its convergence analysis. Section \ref{sec:recovery} establishes the theoretical properties of the algorithm, while Section \ref{sec:experiments} illustrates results on both simulated and real networks. Finally, Section \ref{sec:conclusions} provides conclusions and discussion on future work.
\section{Related works}
\label{sec:relworks}
The method proposed in this paper belongs to the third class of algorithms, which operate directly on the multiplex network model. There are different types of algorithms that fall in this class: random walk, statistical generative network models, label propagation, objective function optimization, and Nonnegative Matrix Factorization (NMF).

Methods based on random walkers model the dynamic process on networks as random walks where the process is
more likely to persist on the vertices in the same community
and far less on the vertices in different communities. For instance, LART \cite{kuncheva2015community} is initialized by assigning each node in each layer to its own community. Hierarchical clustering is then used to merge nodes based on a distance matrix. The partition with the highest multiplex modularity is chosen. 
In \cite{de2015identifying}, Infomap which is based on a compression of network flows is proposed to identify communities within and across layers. However, Infomap tends to assign each physical node across layers to the same community, not differentiating the topological differences across layers. 

Statistical methods such as \cite{ali2019latent} use Weighted Stochastic Block Model (WSBM) to detect common and private communities in heterogeneous weighted networks. Although this method addresses the heterogeneity of networks across layers, the method is limited to detecting only common communities that are shared by all layers, ignoring communities that may be shared by only a subset or different subsets of layers. In \cite{de2017community}, authors propose a generative model and an expectation maximization algorithm for community detection and link prediction in multilayer networks. Although the method allows for different connectivity patterns in each layer, the interdependence between layers is only taken into account for link prediction, while the layers are assumed to share a common community structure.

The third class of methods, Label Propagation Algorithms (LPA), is  based on the intuition that a label can become dominant in a densely connected group of nodes but will have trouble crossing a sparsely connected region. In \cite{boutemine2017mining}, an LPA-based method for community detection in multidimensional networks is proposed to identify communities and the subset of layers in which each of these communities is observed, simultaneously. However, this algorithm fails to detect communities that are private to each layer and communities that may be common among a small number of layers. 

The fourth type of multiplex community detection methods is based on defining an objective function and identifying the community structure that maximizes/minimizes the objective function.  For example, Generalized Louvain (GenLouvain) \cite{mucha_community_2010} uses an extended definition of modularity and is one of the fastest methods for community detection in multiplex networks. As GenLouvain assigns each node-layer tuple to its own community, it cannot identify common communities across layers. More recently, multiobjective genetic and evolutionary algorithms such as MultiMOGA  \cite{amelio2014community} and MOEA/D-TS \cite{karimi2020multiplex} have been used to jointly maximize the modularity of each layer and the similarity between the community structures across layers. These methods find a shared community structure across all layers, not differentiating communities that may be unique to each layer. In \cite{chen2017block},  extension of normalized cut to multiplex networks is proposed by constructing a block Laplacian matrix with each block corresponding to a layer. This method relies on selecting a parameter $\beta$ that controls the consistency of the community structure across different layers.  

The last class of methods is based on NMF which, because of its interpretability and good performance, has been broadly used for community detection in single-layer, multiplex, multilayer, and dynamic networks \cite{ding_equivalence_2005,wang_nonnegative_2011,mankad2013structural,sun2017non}. In \cite{ma2018community},  Semi-Supervised joint Nonnegative Matrix Factorization (S2-jNMF) is proposed for detecting the common communities across layers in a multiplex network. A greedy search of dense subgraphs is performed and these subgraphs are used as {\em a priori} information to create new adjacency matrices for each layer. 
In \cite{gligorijevic2018non}, a two-step approach is proposed, where first  a nonnegative low dimensional feature representation of each layer is found using one of the four different NMF models. These community structures are then used to obtain a consensus community structure. Authors in \cite{nguyen2015community} use NMF for detecting communities in multiplex social networks, where both unifying and coupling approaches are proposed. The unifying approach finds a common community structure  by aggregating all layers, while the coupling approach finds mostly consistent community structures. 
Most of the aforementioned NMF based methods find a common structure across all layers or for a majority of layers and do not consider cases where common communities may be present in different subsets of layers. Moreover, they do not detect private communities. These methods  also require that the number of communities is provided {\em a priori}. 
\section{Background}\label{sec:background}
\subsection{Community Detection}
Community detection for a single-layer network is the partitioning of a node set $V$ as $\mathcal{C} = \{C_1, ..., C_K\}$ where $K$ is the number of communities. 
One of the most popular algorithms for partitioning graphs is the minimum cut method and its variants such as ratio cut and normalized cut \cite{von2007tutorial}. In these methods, the network is partitioned such that the number of edges between different communities is minimized.

Given a single layer graph, $\mathcal{G}=\{V, E, \matr{A}\}$, where $V$, $E$ and $\matr{A}$ are the set of nodes, edges, and adjacency matrix of the graph, respectively, the min-cut problem aims to find a partition by minimizing the following objective function:
\begin{equation}
    \label{Eq: Cut Problem}
    \frac{1}{2}\sum_{p=1}^{K} \mathbf{Cut}(C_p,\bar{C}_p),
\end{equation}
where $\bar{C}_p$ is the complement of $C_p$ and $\mathbf{Cut}(C_p,\bar{C}_p)=\sum_{i\in C_p, j\in \bar{C}_p} A_{ij}$ for  two disjoint sets $C_p$ and $\bar{C}_p$. The min-cut problem is NP-hard. However,  it has been shown \cite{von2007tutorial,ding_equivalence_2005} that  spectral clustering and nonnegative matrix factorization provide  solutions to relaxed versions of the min-cut problem. 

In particular, spectral clustering solves the min-cut problem by embedding each node in a lower
dimensional subspace spanned by the eigenvectors of the normalized Laplacian matrix $\matr{L}$, where $\matr{L}=\matr{D}^{-1/2}(\matr{D-A})\matr{D}^{-1/2}$,
 with $\matr{D}$ being the diagonal matrix defined as $D_{ii}=\sum_{j} A_{ij}$. k-means clustering is then applied to this low-dimensional subspace spanned by the   eigenvectors. Authors in \cite{ding_equivalence_2005} show that NMF is equivalent to Laplacian based spectral clustering. Normalized Cut using the normalized adjacency matrix, $\tilde{\matr{A}}=\matr{D}^{-1/2}\matr{AD}^{-1/2}$, is equivalent to the nonnegative matrix factorization problem 
$\underset{\matr{H}\geq 0}{\operatorname{argmin}}=||\tilde{\matr{A}}-\matr{HH}^\top||^2$. 

\subsection{Multiplex Networks}
Multiplex networks can be represented using a finite sequence of graphs $\{\mathcal{G}_l\}$, where $l \in \{1,2,\ldots,L\}$, $\mathcal{G}_l=(V_l,E_l, \matr{A}_{l})$ \cite{cozzo2015structure}. $V_l$ is the set of nodes in layer $l$  and $\matr{A}_l\in \mathbb{R}^{n\times n}$ is the adjacency matrix for layer $l$. In this paper, we use undirected (symmetric) weighted and binary adjacency matrices. For a weighted adjacency matrix, $A_{l_{ij}}\in [0,1]$, and for a binary adjacency matrix, $A_{l_{ij}}\in\{0,1\}$.

\subsection{Orthogonal Nonnegative Matrix Tri-Factorization}
Nonnegative Matrix Factorization decomposes a nonnegative matrix $\matr{W}\in \mathbb{R}^{n\times m}$ into the product of two low-rank nonnegative matrices $\matr{V}\in \mathbb{R}^{n\times k}$ and $\matr{U}\in \mathbb{R}^{m\times k}$, such that $\matr{W}\approx \matr{VU}^\top$ and $k \ll n,m$. $\matr{V}$ and $\matr{U}$ are found by solving the optimization problem
\begin{equation}
\label{eq:NMF}
\centering
    \underset{\textbf{V},\textbf{U}>0}{\operatorname{argmin}}||\matr{W}-\matr{VU}^\top||_{F}^2.
\end{equation}
In NMF-based community detection, $\matr{V}$ and $\matr{U}$ are the community feature matrix
 and the community indicator matrix, respectively.
Adding orthogonality constraints 
improves the performance of NMF as orthogonality and nonnegativity force each row of $\matr{V}$ ($\matr{U}$) to have only one nonzero element which implies that each node belongs only to one community. Orthogonal NMF has been broadly used in community detection where $\matr{W}$ is usually the adjacency matrix and $k$ is the number of communities \cite{wu2018nonnegative,lu2020community}.

\begin{figure*}[t]
    \centering
    \includegraphics[width=0.95\linewidth]{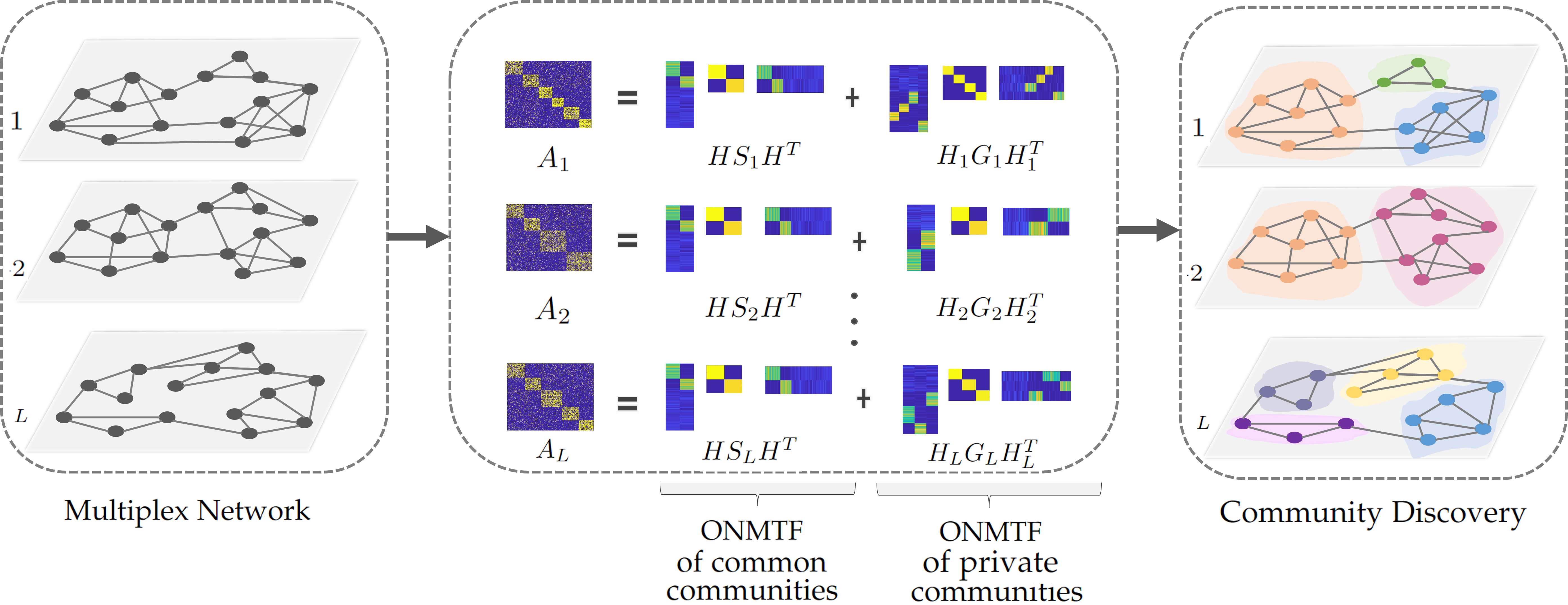}
    \caption{Illustration of the proposed community detection algorithm.}
    \label{fig:factLl}
\end{figure*}

Different extensions of NMF, such as Symmetric Nonnegative Matrix Tri-factorization \cite{ding_orthogonal_2006}, where $\matr{W}$ is approximated by $\matr{USU}^\top$, have been proposed for community detection. $\matr{U}$ contains community membership information while the symmetric matrix $\matr{S}\in \mathbb{R}^{k \times k}$ provides more degrees of freedom to the approximation. In this paper, we will use Orthogonal Nonnegative Matrix Tri-factorization (ONMTF) to formulate the community detection problem in multiplex networks.
\section{Proposed Method (MX-ONMTF)}\label{sec:mxonmtf}
The proposed method, MX-ONMTF, models each layer’s adjacency matrix as a sum of low-rank representations of common and private communities using Orthogonal Nonnegative Matrix Tri-Factorization (ONMTF). Fig. \ref{fig:factLl} illustrates the overview of the proposed algorithm for a multiplex network with $L$ layers and two common communities.

\subsection{Problem Formulation}
In this paper, we define common communities as communities that are observed in more than one layer.\\

\noindent \textbf{Definition 1:} An ideal common community $C$ in a multiplex network $\{\mathcal{G}_l\}$ is defined as a subgraph with the same set of nodes for a subset of layers $\textbf{m} \subseteq \{1,2,\ldots,L\}$, where $|\textbf{m}|> 1$. Mathematically, $C$ can be defined as\\
\begin{align*}
\begin{split}
    C=&\{(V^C_l, E^C_{l}): V^C_l \subseteq V_l, E^C_{l} \subseteq  V^C_l \times V^C_l, V_{l}^{C}=V_{k}^{C} \\
    & l, k \in \textbf{m}, 
    \textbf{m} \subseteq \{1,2,\ldots,L\}, |\textbf{m}|> 1\}.\\
\end{split}
\end{align*}

\noindent \textbf{Definition 2:} A private community $C$ in a multiplex network $\{\mathcal{G}_l\}$ is defined as any community that is not common across at least two layers. 

For a  multiplex network with $L$ layers and adjacency matrices, $\matr{A}_{l} \in \mathbb{R}^{n\times n}$, $l\in \{1,2,\ldots,L\}$, we model each layer's adjacency matrix in terms of common and individual communities using ONMTF. The resulting objective function can be formulated as

\begin{equation}
\begin{aligned}
\label{eq:ObjFn}
    \underset{\matr{H}\geq0,\matr{H}_{l}\geq0,\matr{S}_{l}\geq0,\matr{G}_{l}\geq0}{\operatorname{argmin}} \sum_{l=1}^{L} || \matr{A}_{l}-\matr{HS}_{l}\matr{H}^\top-\matr{H}_{l}\matr{G}_{l}\matr{H}_{l}^\top ||^2_{F}\\
    \mbox{s.t\hspace{0.1cm}} \matr{H}^\top\matr{H}=\matr{I}, \matr{H}_{l}^\top\matr{H}_{l}=\matr{I}, \mbox{with \hspace{0.1cm}} l\in\{1,2,\ldots,L\},
\end{aligned}
\end{equation}
    
\noindent where $\matr{H}\in \mathbb{R}^{n\times k_{c}}$  and $\matr{H}_{l}\in \mathbb{R}^{n\times k_{p_l}},l\in\{1,2,\ldots,L\}$ are the community membership matrices corresponding to the common and private communities, respectively, and $\matr{S}_l$ and $\matr{G}_l$ are symmetric matrices.  In this work, it is assumed that the $L$ layers have a total of $k_{c}$ common communities and $k_{p_l}$ private communities in each layer $l$. The goal is to simultaneously identify communities that are common across any subset of two or more layers and communities that are unique to each layer.
Therefore, $\matr{H}$ will contain  information for all  common communities.

\subsection{Optimization solution}
ONMTF optimization problem in \eqref{eq:ObjFn} can be solved using a multiplicative update algorithm (MUA) \cite{ding_orthogonal_2006}. Multiplicative update algorithms for solving NMF problems were introduced in \cite{lee_algorithms_2001}, while solving NMTF with orthogonal constraints was first addressed by \cite{ding_orthogonal_2006}. In this paper, we follow their approach to derive the multiplicative update rules for each variable. 

To find the update rules for $\matr{H}$, $\matr{H}_l$, $\matr{S}_l$, and $\matr{G}_l$, the following Lagrangian function with Lagrange multipliers $\matr{\Lambda}$ and $\matr{\Lambda}_l$ is minimized:
\begin{align}
\begin{split}
    \label{eq:Lagrangian}
   &\mathcal{L}(\matr{H}, \matr{H}_l,\matr{S}_l, \matr{G}_l) =\sum_{l=1}^{L} || \matr{A}_{l}-\matr{HS}_{l}\matr{H}^\top-\matr{H}_{l}\matr{G}_{l}\matr{H}_{l}^\top ||^2_{F} 
   \\
   &+ tr(\matr{\Lambda(H}^\top\matr{H - I}))+\sum_{l=1}^{L}tr(\matr{\Lambda}_l(\matr{H}_{l}^\top\matr{H}_l - \matr{I})).
\end{split}
\end{align}
For updating $\matr{H}$, we find $\matr{\nabla_H \mathcal{L}}$ as
\begin{align}
\label{eq:gradient1}
    \begin{split}
        \matr{\nabla_H \mathcal{L}}= &\sum_{l=1}^{L}(\matr{4HS}_{l}^\top\matr{H}^\top\matr{HS}_{l}  + 4\matr{H}_{l}\matr{G}_{l}^\top\matr{H}_{l}^\top\matr{HS}_{l}- 4\matr{A}_{l}\matr{HS}_{l})\\+  &\matr{4H\Lambda}.
    \end{split}
\end{align}
\vspace{-0.1cm}
Applying the KKT conditions $\matr{\nabla_H \mathcal{L}}= 0$ and $\matr{\nabla_\Lambda \mathcal{L}}=0$, we  obtain:\\
(i) $\matr{\Lambda}= \sum_{l=1}^{L}(\matr{-S}_{l}^\top\matr{S}_{l}  - \matr{H}^\top\matr{H}_{l}\matr{G}_{l}^\top\matr{H}_{l}^\top\matr{HS}_{l} + \matr{H}^\top\matr{A}_{l}\matr{HS}_{l})$.\\
(ii)  $\matr{H}^\top\matr{H=I}$.\\
Substituting (i) and (ii) in  Eq. \eqref{eq:gradient1}, we get
\begin{align}
    \begin{split}
    \label{eq:gradient2}
        \matr{\nabla_H \mathcal{L}} =& \sum_{l=1}^{L}(\matr{4H}_{l}\matr{G}_{l}^\top\matr{H}_{l}^\top\matr{HS}_{l} - 4\matr{A}_{l}\matr{HS}_{l} + 4\matr{HH}^\top\matr{A}_{l}\matr{HS}_{l}\\\matr{-4H H}^\top&\matr{H}_{l}\matr{G}_{l}^\top\matr{H}_{l}^\top\matr{HS}_{l}).
    \end{split}
\end{align}

As discussed in \cite{yoo_orthogonal_2010}, if the gradient of an error function, $\matr{\varepsilon}$, is of the form $\matr{\nabla \varepsilon= \nabla \varepsilon^+ -\nabla \varepsilon^-} $, where $\matr{\nabla \varepsilon^+>0}$ and $\matr{\nabla \varepsilon^->0}$, then the multiplicative update for parameter $\matr{\Theta}$ has the form $\matr{\Theta=\Theta\odot\frac{\nabla \varepsilon^-}{\nabla \varepsilon^+}}$. It can be easily seen that the multiplicative update  preserves the nonnegativity of $\matr{\Theta}$, while $\matr{\nabla \varepsilon}=0$ when the convergence is achieved. Following this procedure, from the gradient of the error function in Eq. \eqref{eq:gradient2}, we derive the following multiplicative update rule for $\matr{H}$

\begin{equation}
\label{eq:Hupdate}
        \matr{H} \leftarrow \matr{H}\odot\frac{\sum_{l=1}^{L}(\matr{A}_{l}\matr{HS}_{l}+ \matr{HH}^\top\matr{H}_{l}\matr{G}_{l}^\top\matr{H}_{l}^\top\matr{HS}_{l})}{\sum_{l=1}^{L}(\matr{H}_{l}\matr{G}_{l}^\top\matr{H}_{l}^\top\matr{HS}_{l} + \matr{HH}^\top\matr{A}_{l}\matr{HS}_{l})},
\end{equation}
\noindent where the multiplication and division are performed element-wise and both numerator and denominator are positive. Similarly, we obtain the following update rules for $\matr{H}_l$, $\matr{S}_l$, and $\matr{G}_l$, for each $l\in\{1,2,\ldots,L\}$:
\begin{equation}
\label{eq:Hlupdate}
        \matr{H}_l \leftarrow \matr{H}_l\odot\frac{\matr{A}_{l}\matr{H}_{l}\matr{G}_{l} + \matr{H}_{l}\matr{H}_{l}^\top\matr{HS}_{l}^\top\matr{H}^\top\matr{H}_{l}\matr{G}_{l}}{\matr{HS}_{l}^\top\matr{H}^\top\matr{H}_{l}\matr{G}_{l}^\top + \matr{H}_{l}\matr{H}_{l}^\top\matr{A}_{l}\matr{H}_{l}\matr{G}_{l}},\\
\end{equation}

\begin{equation}
\label{eq:Supdate}
    \matr{S}_l \leftarrow \matr{S}_l\odot\frac{\matr{H}^\top\matr{A}_{l}\matr{H}}{\matr{H}^\top\matr{HS}_{l}\matr{H}^\top\matr{H} + \matr{H}^\top\matr{H}_{l}\matr{G}_{l}\matr{H}_{l}^\top\matr{H}},
\end{equation}

\begin{equation}
\label{eq:Gupdate}
    \matr{G}_l\leftarrow \matr{G}_l\odot\frac{\matr{H}_{l}^\top\matr{A}_{l}\matr{H}_{l}}{\matr{H}_{l}^\top\matr{H}_{l}\matr{G}_{l}\matr{H}_{l}^\top\matr{H}_{l}+\matr{H}_{l}^\top\matr{H}\matr{S}_{l}\matr{H}^\top\matr{H}_{l}}.
\end{equation}

Since NMF algorithms are initialized with random matrices, different runs yield local minima. For this reason, we run the algorithm 50 times and report the best results
\cite{li2018nonnegative,luo2021symmetric}. 
As shown in Algorithm \ref{alg:mx-onmtf}, for each random initialization of $\matr{H}$, $\matr{H}_l$, $\matr{S}_l$, and $\matr{G}_l$, the multiplicative update rules described in   Eqs. \eqref{eq:Hupdate}-\eqref{eq:Gupdate} are repeated for 1000 iterations or until  convergence. We then select the solution that yields the maximum value of the performance metric across the different runs. For synthetic networks for which a ground truth is available, Normalized Mutual Information (NMI) \cite{danon2005comparing} is used.  For networks without ground truth, Modularity Density ($Q_D$) \cite{li2008quantitative} is used as the performance metric.

\begin{algorithm}[h!]
\caption{MX-ONMTF}
\label{alg:mx-onmtf}
\begin{algorithmic}[1]
\renewcommand{\algorithmicrequire}{\textbf{Input:}}
\REQUIRE Adjacency matrices $\matr{A}_l$, $l\in\{1,2,\ldots,L\}$. \\
\renewcommand{\algorithmicrequire}{\textbf{Output:}} 
\REQUIRE Community membership matrices $\matr{H}$, $\matr{H}_l$.
\STATE Use Algorithm \ref{alg:find-kc} to find $k_c$ and $k_{p_l}$.
\FOR{r=1 to 50}
\STATE Randomly initialize $\matr{H}, \matr{H}_l, \matr{S}_l, \matr{G}_l >0$
\FOR{1000 iterations or until convergence}
\STATE update $\matr{H}$ according to Eq. \eqref{eq:Hupdate}
\STATE update $\matr{H}_l$ for each $l\in\{1,2,\ldots,L\}$ according to Eq. \eqref{eq:Hlupdate}
\STATE update $\matr{S}_l$ for each $l\in\{1,2,\ldots,L\}$ according to Eq. \eqref{eq:Supdate}
\STATE update $\matr{G}_l$ for each $l\in\{1,2,\ldots,L\}$ according to Eq. \eqref{eq:Gupdate}
\ENDFOR
\FOR{each layer $l$}
\STATE Apply Algorithm \ref{alg:readH} with $\matr{A}_l$, $\matr{H}$, and $k_{p_l}$ as inputs to  to find $\matr{H}_{c_l}$.
\FOR{each i}
\STATE $j^*\leftarrow argmax_j\matr{H}_{c_l}(i,j)$
\IF{$\matr{H}_{c_l}(i,j^*)>\matr{H}_{l}(i,j^*)$}
 \STATE $idx(i)\leftarrow j^*$ 
\ELSE
\STATE $idx(i)\leftarrow (argmax_j \matr{H}_l(i,j))+ k_c + \sum_{n=1}^{l-1} k_{p_n}$ 
\ENDIF
\ENDFOR
\ENDFOR
\STATE Compute $\text{NMI}_r$ or $Q_{D_r}$.
\ENDFOR
\STATE Choose the partition $r^*=argmax_{r} \text{NMI}_r$\\ ($r^*=argmax_{r} Q_{D_r}$ ).
\end{algorithmic}
\end{algorithm}

\subsection{Number of communities}\label{findk}
In most NMF-based community detection algorithms, the number of communities ($k$) is an input parameter. This problem is usually addressed by detecting communities with different values of $k$ and selecting the one that gives the solution with the best pre-determined performance metric, such as modularity \cite{pramanik2017discovering}.

\begin{algorithm}[h]
\caption{Finding $k_l$, $k_c$ and $k_{p_l}$, for $l\in\{1,2,\ldots,L\}$ .}
\label{alg:find-kc}
\begin{algorithmic}[1]
\renewcommand{\algorithmicrequire}{\textbf{Input:}}
\REQUIRE Adjacency matrices $\matr{A}_l$ for $l\in\{1,2,\ldots,L\}$.\\
\renewcommand{\algorithmicrequire}{\textbf{Output:}} 
\REQUIRE Number of common communities $k_c$, total number of communities per layer $k_l$, and number of private communities per layer $k_{p_l}$, for $l\in\{1,2,\ldots,L\}$.\\
\STATE Let $\matr{L}^{null}_{l}$ be the normalized Laplacian the of Erd\H{o}s–R{\'e}nyi null model. 
\STATE $\matr{L}^{null}_{l}=\matr{V}_{l}\matr{\Lambda}_{l}\matr{V}_{l}^{T}$
\STATE $\delta \leftarrow quantile_{0.95}[max\{|\lambda^{null}_i| - |\lambda^{null}_{i+1}|, i \ge 2\}]$ 
\STATE $k_l\leftarrow min\{k: |\lambda_i|-|\lambda_{i+1}| > \delta , \forall i > k\}$
\STATE Randomly initialize $\matr{U}_l \ge 0, \matr{S}_l \ge 0$.
\FOR{1000 iterations or until convergence}
\STATE update $\matr{U}_l$ using $\matr{U}_l=\matr{U}_l*\frac{(\matr{A}_l\matr{U}_l\matr{S}_l)_{ij}}{(\matr{U}_l\matr{U}_l^\top\matr{A}_l\matr{U}_l\matr{S}_l)_{ij}}$ 
\STATE update $\matr{S}_l$ using $\matr{U}_l=\matr{U}_l*\frac{(\matr{U}_l^\top\matr{A}_l\matr{U}_l)_{ij}}{(\matr{U}_l^\top\matr{U}_l\matr{S}_l\matr{U}_l^\top\matr{U}_l)_{ij}}$
\ENDFOR
\STATE $\matr{X}=[\matr{U}_1,\matr{U}_2,\ldots,\matr{U}_L]^\top \in \mathbb{R}^{m\times n}$, $m=\sum_{l=1}^L k_l$
\STATE $\matr{F}\leftarrow AgglomerativeHierarchicalClustering(\matr{X})$
\STATE $k_c=0$
\FOR{$i=2$ to $m-2$} 
\STATE $d_i\leftarrow \frac{\matr{F}(i,3)-\matr{F}(i-1,3)}{\matr{F}(i-1,3)}$
\IF{$d_i\ge 0.5$}
\IF{$max\{\matr{F}(i,1),\matr{F}(i,2)\} \le m$}
\STATE $k_{c}\leftarrow k_{c}+1$
\ELSE 
\STATE $k_{c}\leftarrow k_{c}$
\ENDIF 
\ELSE
\STATE $cut\leftarrow i$ and stop \textbf{for}
\ENDIF 
\ENDFOR
\STATE $C \leftarrow$find$(1+\sum_{j=1}^{l-1} k_j \le \matr{F}(1:cut,1:2) \le \sum_{j=1}^{l}k_j)$
\STATE $k_{p_l}\leftarrow k_l - |C|$, for $l\in\{1,2,\ldots,L\}$\\
\end{algorithmic}
\end{algorithm}

In this paper, a two-step approach is proposed to determine the number of communities per layer and the number of common communities. First, the number of communities per layer ($k_1$, $k_2$,\ldots, $k_L$), is found  using the eigengap rule \cite{liu2018global}. Next, ONMTF is applied  to each layer \cite{ding_orthogonal_2006} and  the low-rank embedding matrices, $\matr{U}_l\in \mathbb{R}^{n\times k_l}$, are obtained. Each element of the embedding matrices, $\matr{U}_l(i,j)$, represents the likelihood of node $i$  belonging to community $j$.  An agglomerative hierarchical clustering algorithm using  Euclidean distance is applied on the  rows of $X=[\matr{U}_1,\matr{U}_2,\ldots, \matr{U}_L]^\top\in \mathbb{R}^{m\times n}$, where $m=\sum_{l=1}^L k_l$, to obtain the number of common communities. At each
step of the algorithm, the two columns with the smallest distance are aggregated, and the distances between the newly formed cluster and the remaining ones are updated. A dendrogram like the one shown in Fig. \ref{fig:dend} can be used to represent the different iterations of this algorithm. The $m$ leaves of the dendrogram correspond to the total number of communities across the $L$ layers. 

This agglomerative hierarchical clustering algorithm outputs  a matrix $\matr{F}\in \mathbb{R}^{m-1 \times 3}$. The first two columns of $\matr{F}(i,:)$ correspond to the labels of the two leaves of the dendrogram that form cluster $m + i$ and the third column contains the distance between these two leaves. This distance matrix $\matr{F}$ is used in Algorithm \ref{alg:find-kc} to determine the number of common communities, $k_c$, and the number of private communities per layer, $k_{p_{l}}$. The algorithm iterates until the minimum distance between any two clusters increases by more than $50\%$ of the minimum distance from the previous iteration. Fig. \ref{fig:dend} shows the dendrogram of the hierarchical clustering of the columns of the embedding matrices $\matr{U}_1,\matr{U}_2,\matr{U}_3$ of a 3-layer network with $k_1=6$, $k_{2}=6$, and $k_{3}=5$ and the red line indicates where the algorithm stops. For this example, $k_c=3$, $k_{p_1}=3$, $k_{p_2}=4$, and $k_{p_3}=3$. 

\begin{figure}[h!]
            \centering
           \includegraphics[width=0.95\linewidth]{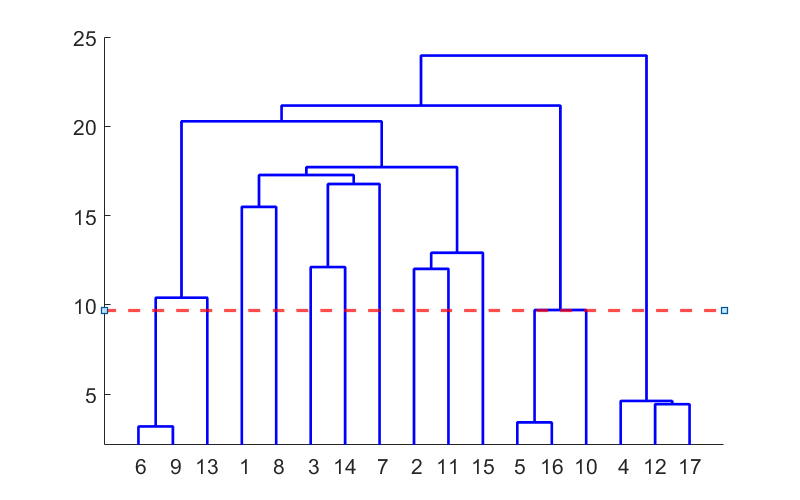}
            \caption{Example dendrogram illustrating the hierarchical clustering of the columns of the embedding matrices $\matr{U}_1,\matr{U}_2,\matr{U}_3$ for a 3-layer network with 3 common communities. The red line indicates where the algorithm stops.}
            \label{fig:dend}
\end{figure}

\subsection{Determining the common community labels for each layer} \label{ReadingH}
$\matr{H}\in \mathbb{R}^{n\times k_{c}}$  is the community membership matrix corresponding to the common communities. Each row of $\matr{H}$ indicates whether a particular node belongs to any of the $k_{c}$ common communities.  
Algorithm \ref{alg:readH} describes how the membership of the common communities is determined across layers, i.e., which columns of $\matr{H}$ contain information about the common communities present in layer $l$.

\begin{algorithm}[h!]
\caption{Identify the membership of $k_c$ common communities across layers $\{1,2,\ldots,L\}$}
\label{alg:readH}
\begin{algorithmic}[1]
\renewcommand{\algorithmicrequire}{\textbf{Input:}}
\REQUIRE Community membership matrices $\matr{H}, \matr{H}_l$, adjacency matrices $\matr{A}_l$, and number of private communities $k_{p_l}$, for each layer $l\in\{1,2,\ldots,L\}$.\\
\renewcommand{\algorithmicrequire}{\textbf{Output:}} 
\REQUIRE Layer specific community membership matrix $\matr{H}_{c_l}$ containing information about the common communities present in layer $l$. 
\FOR{$l=1$ to $L$}
\FOR{each node $i$}
\STATE $j^*\leftarrow argmax_j(\matr{[H,H}_l])_{ij}$
\STATE $idx(i)\leftarrow j^*$
\ENDFOR
\FOR{$j=1$ to $k_c$}
\STATE $\matr{H}(find(idx==j),j)=1$
\STATE $\matr{H}(find(idx\ne j),j)=0$
\ENDFOR
\FOR{$j=1$ to $k_c$}
\STATE $m=|\matr{H}(:,j)==1|$
\STATE $\matr{B}_0\leftarrow \matr{A}_l	\odot(\matr{H}(:,j)\matr{H}(:,j)^\top)$
\STATE $\matr{B}_1\leftarrow \matr{A}_l	\odot(\matr{1}_{n \times n}-\matr{H}(:,j)\matr{H}(:,j)^\top)$ 
\STATE $q(j)\leftarrow \frac{\sum_{v,w}(\matr{B}_1)_{vw}/(m(m-1))}{\sum_{v,w}(\matr{B}_0)_{vw}/(m(n-m))}$
\ENDFOR
\STATE $[q_{sorted}, J]\leftarrow sort($q$,descending)$. $J$ contains the sorted indices of the elements of $q$.
\STATE $pos\leftarrow J(1:k_{p_l})$.
\STATE $\matr{H}_{c_l}\leftarrow \matr{H}(:,[pos])$
\ENDFOR
\end{algorithmic}
\end{algorithm}

\subsection{Time complexity}
The time complexity of the proposed algorithm is mostly due  to the Multiplicative Updates Rules, Eqs. \eqref{eq:Hupdate}-\eqref{eq:Gupdate}. The time complexity
for the product of two matrices, e.g., the product of a $m \times k$ matrix by
a $k \times n$ matrix, is $\mathcal{O}(mkn)$. Table \ref{tab:complexity} shows 
the time complexities of Eqs. \eqref{eq:Hupdate}-\eqref{eq:Gupdate} and the total complexity, with $l\in\{1,2,\ldots,L\}$, and $k=k_c+ \sum_{l=1}^{L}k_{p_l}$.

\begin{table}[h!]
    \centering
    \begin{tabular}{cc}
    \hline
        $\matr{H}$ update (Eq. \eqref{eq:Hupdate}) & $\mathcal{O}(n^2(k_c+ \sum_{l=1}^{L}k_{p_l}))$\\
        $\matr{H}_l$ update (Eq. \eqref{eq:Hlupdate}) & $\mathcal{O}(n^2(k_c + k_{p_l}))$ \\
        $\matr{S}_l$ update (Eq. \eqref{eq:Supdate}) & $\mathcal{O}(n^2k_c)$\\
        $\matr{G}_l$ update (Eq. \eqref{eq:Gupdate}) & $\mathcal{O}(n^2k_{p_l})$\\
        \hline
        Total & $\mathcal{O}(n^2k)$\\
        \hline
        \end{tabular}
    \caption{Computational complexity of updating each variable per iteration.}
    \label{tab:complexity}
\end{table}

\vspace{-0.5cm}
\subsection{Storage complexity}

The storage complexity of our algorithm is determined by the sizes of the matrices $\matr{H}$, $\matr{H}_l$, $\matr{S}_l$, and $\matr{G}_l$. It can be seen that the total storage complexity is $\mathcal{O}(nk)$. For a multiplex network of size $n \times n \times L$, this is a significant reduction in memory cost.
\begin{table}[h]
    \centering
    \begin{tabular}{cc}
    \hline
        $\matr{H}$  & $\mathcal{O}(nk_c)$\\
        $\matr{H}_l$ & $\mathcal{O}(nk_{p_{l}})$ \\
        $\matr{S}_l$ & $\mathcal{O}({k_c}^2)$\\
        $\matr{G}_l$ & $\mathcal{O}({k_{p_{l}}}^2)$\\
        \hline
        Total & $\mathcal{O}(nk)$\\
        \hline
        \end{tabular}
    \caption{Storage complexity of each variable.}
    \label{tab:complexity}
\end{table}

\section{Convergence Analysis}\label{sec:convergence}

In this section, we will prove the convergence of the multiplicative update rule defined by Eq. \eqref{eq:Hupdate}  using the auxiliary function approach.  As the other update rules are similar, we will not explicitly prove their convergence. We first introduce the definition of auxiliary function as follows.

\noindent \textbf{Definition 1:} A function $Z(\matr{H},\matr{H}^t)$
is called an auxiliary function of $\matr{\mathcal{L}(H)}$ if it satisfies
\begin{center}
  $Z(\matr{H,H}^t)\ge \matr{\mathcal{L}(H)}$ and $Z(\matr{H,H})=\matr{\mathcal{L}(H)}$.
\end{center}
The auxiliary function is a useful concept because of the following lemma which is proved in \cite{lee_algorithms_2001}.

\begin{lemma}
If $Z$ is an auxiliary function, then $\mathcal{L}$ is non-increasing under the update
\begin{center}
    $\matr{H}^{t+1}=\underset{\matr{H}}{\operatorname{argmin}} Z(\matr{H,H}^t)$.
\end{center}
\end{lemma}

\begin{theorem}
Given $\matr{H}_l$, $\matr{S}_l$, and $\matr{G}_l$ the Lagrangian function $\matr{\mathcal{L}(H)}$ is monotonically decreasing under the update rule \eqref{eq:Hupdate}. 
\end{theorem}

\begin{proof}
For convenience, let $\mathcal{L}(h)$ denote the part of  $\mathcal{L}(\matr{H})$ dependent on $H_{ij}$. From Eq. \eqref{eq:gradient2}  we have 
\vspace{-0.1cm}
\begin{align*}
\begin{split}
    &\mathcal{L'}(h)=\sum_{l=1}^{L}(\matr{4H}_{l}\matr{G}_{l}^\top\matr{H}_{l}^\top\matr{HS}_{l} - 4\matr{A}_{l}\matr{HS}_{l}\\ &+ 4\matr{HH}^\top\matr{A}_{l}\matr{HS}_{l}\matr{-4H H}^\top\matr{H}_{l}\matr{G}_{l}^\top\matr{H}_{l}^\top\matr{HS}_{l})_{ij}.\\
\end{split}
\end{align*}
The second-order derivative of $\mathcal{L}(h)$ with respect to $h_{ij}$ is
\begin{align*}
\begin{split}
\mathcal{L''}(h)=\sum_{l=1}^{L}\{4(\matr{H}_l\matr{G}_l\matr{H}_l^\top)_{ii}\matr{S}_{l_{jj}}-4\matr{A}_{l_{ii}}\matr{S}_{l_{jj}} \\+ 4[(\matr{H}^\top\matr{A}_l\matr{HS}_l)_{ij}+h_{ij}(\matr{A}_l\matr{HS}_l)_{ij}+(\matr{HH}^\top\matr{A}_l)_{ii}\matr{S}_{l_{jj}}]\\ 
-4[(\matr{H}^\top\matr{H}_l\matr{G}_l\matr{H}_l^\top\matr{HS}_l)_{ij}+h_{ij}(\matr{H}_l\matr{G}_l\matr{H}_l^\top\matr{HS}_l)_{ij}\\+(\matr{HH}^\top\matr{H}_l\matr{G}_l\matr{H}_l^\top)_{ii}\matr{S}_{l_{jj}}]\}.
\end{split}
\end{align*}

Let $h^t_{ij}$ denote the updated value of $h_{ij}$ after the $t$th iteration, then the Taylor series expansion of $\mathcal{L}(h)$ at $h^t_{ij}$ can be written as
\begin{align*}
\begin{split}
\mathcal{L}(h)=\mathcal{L}(h^t_{ij})+\mathcal{L'}(h^t_{ij})(h-h^t_{ij})+ \frac{1}{2}\mathcal{L''}(h_{ij}^t)(h-h^t_{ij})^2.
\end{split}
\end{align*}
Now, the key is to find an appropriate auxiliary function $Z(h,h_{ij}^t)$. We choose the following $Z(h,h_{ij}^t)$ and prove in Appendix A, that it satisfies the conditions to be an auxiliary function of $\mathcal{L}(h)$. 

\begin{align}
    \begin{split}
    &Z(h,h_{ij}^t)=\mathcal{L}(h^t_{ij})+3\mathcal{L'}(h^t_{ij})(h-h^t_{ij})+ \\
    &\frac{3}{2}\frac{\sum \limits_{l=1}^{L}(4\matr{H}^t_l\matr{G}_l{\matr{H}_l^t}^\top\matr{H}^t\matr{S}_l+4\matr{H}^t{\matr{H}^t}^\top\matr{A}_l\matr{H}^t\matr{S}_l)_{ij}}{h_{ij}^t}(h-h^t_{ij})^2.
    \end{split}
    \label{eq:auxiliary}
\end{align}

According to Lemma 1, we must find the minimum of $Z(h,h_{ij}^t)$ with respect to $h$.\\

\begin{align*}
    \frac{\partial Z(h,h_{ij}^t)}{\partial h}=&3\mathcal{L'}(h^t_{ij})
    +3\frac{\sum_{l=1}^{L}(4\matr{H}_l\matr{G}_l\matr{H}_l^\top\matr{H}^t\matr{S}_l)_{ij}}{h_{ij}^t}(h-h^t_{ij})\\ +&3\frac{\sum_{l=1}^{L}(4\matr{H}^t{\matr{H}}^{t^\top}\matr{A}_l\matr{H}^t\matr{S}_l)_{ij}}{h_{ij}^t}(h-h^t_{ij})=0
\end{align*}

Replacing $\mathcal{L'}(h^t_{ij})$ in the equation above and canceling the common terms, we obtain

\begin{align*}
    \begin{split}
        \sum_{l=1}^{L}(-4\matr{A}_{l}\matr{H}^t\matr{S}_{l} \matr{-4H}^t{\matr{H}}^{t^\top}\matr{H}_{l}\matr{G}_{l}^\top\matr{H}_{l}^\top\matr{H}^t\matr{S}_{l})_{ij} \\+ \sum_{l=1}^{L}(4\matr{H}_l\matr{G}_l\matr{H}_l^\top\matr{H}^t\matr{S}_l+4\matr{H}^t{\matr{H}}^{t^\top}\matr{A}_l\matr{H}^t\matr{S}_l)_{ij}\frac{h}{h_{ij}^t}=0.
    \end{split}
\end{align*}

Replacing $h$ by $h_{ij}^{t+1}$ we obtain the following update rule \\

\begin{align*}
    h^{t+1}_{ij}=h^t_{ij}\frac{\sum_{l=1}^L(\matr{A}_l\matr{H}^t\matr{S}_l+\matr{H}^t{\matr{H}}^{t^\top}\matr{H}_l\matr{G}_l^\top\matr{H}_l^\top\matr{H}^t\matr{S}_l)_{ij}}{\sum_{l=1}^L(\matr{H}_l\matr{G}_l^\top\matr{H}_l^\top\matr{H}^t\matr{S}_l+\matr{H}^t{\matr{H}}^{t^\top}\matr{A}_l\matr{H}^t\matr{S}_l)_{ij}},
\end{align*}

which is the same as the update rule shown in Eq. \eqref{eq:Hupdate}.
\end{proof}

\section{Recovery Guarantees}\label{sec:recovery}


In this section, we will establish the theoretical properties of the proposed community detection method. We want to determine if, by optimizing the objective function, our algorithm will return good community structures. In particular, we will investigate the consistency properties of the global optimizer of the objective function under the multilayer stochastic blockmodel (MLSBM).
The optimization problem in \eqref{eq:ObjFn} can be rewritten as

\begin{equation}
\begin{aligned}
\label{eq:ObjFnEq}
    \underset{\matr{H'}_{l},\matr{F}_{l}}{\operatorname{argmin}} \sum_{l=1}^{L} || \matr{A}_{l}-\matr{H'}_{l}\matr{F}_{l}\matr{H'}_{l}^\top ||^2_{F}
    \mbox{\hspace{0.2cm}s.t\hspace{0.2cm}} \matr{H'}_{l}^\top\matr{H'}_{l}=\matr{I}, 
\end{aligned}
\end{equation}

\noindent where $\matr{F}_l$ is a block matrix defined as 

\[ \matr{F}_l=
\begin{bmatrix}
\matr{S}_l & \bf{0}\\
\bf{0} & \matr{G}_l
\end{bmatrix}
\]
and $\matr{H'}_l=[\matr{H}|\matr{H}_l]$ is the concatenation of the community membership matrices of the common and private communities, $\matr{H}\in \mathbb{R}^{n\times kc}$ and $\matr{H}_{l}\in \mathbb{R}^{n\times k_{p_l}}$, respectively. 

For the  $n\times n\times L$ population adjacency
tensor $\mathcal{A} = \{\matr{A}_1,...,\matr{A}_L\}$, we can define a multiplex SBM $[\matr{Z},\Theta]$ as in \cite{ali2019latent}, with each of the $L$ slices $\matr{A}_l\in \mathbb{R}^{n\times n}$. The mulitplex SBM with parameters $[\mathcal{Z}=\{\matr{Z}_1,...,\matr{Z}_L\},\mathcal{B} = \{\matr{\Theta}_1,...,\matr{\Theta}_L\}]$, can be
written in the matrix form as,

\begin{equation*}
     \mathbb{E}(\matr{A}_l)= \matr{Z}_l\matr{\Theta}_l\matr{Z}_l^\top, 
\end{equation*}

\noindent with $\matr{\Theta}_l\in [0,1]^{(kc+k_{p_l})\times(kc+k_{p_l})}$ and $\matr{Z}_l\in \{0,1\}^{n\times (kc+k_{p_l})}$ for each layer $l$.  $\matr{\Theta}_l$ is a block matrix defined as 

\[ \matr{\Theta}_l=
\begin{bmatrix}
\matr{\Theta}_c & \bf{0}\\
\bf{0} & \matr{\Theta}_{p_l}
\end{bmatrix},
\]

\noindent with $\matr{\Theta}_c$ and $\matr{\Theta}_{p_l}$ being the affinity probability matrices of the common and private communities, respectively.  $\matr{Z}_l=[\matr{Z}_c|\matr{Z}_{p_l}]$ is the concatenation of the community membership matrices of the common and private communities, $\matr{Z}_c\in \mathbb{R}^{n\times kc}$ and $\matr{Z}_{p_l}\in \mathbb{R}^{n\times k_{p_l}}$, respectively.  

To prove that our method can correctly recover the community assignments, we propose the following lemma following the work in \cite{paul2020spectral}.

\begin{lemma}
The optimization problem in \eqref{eq:ObjFn} applied to $\mathcal{A}$
has $\matr{H'}_l= \matr{Z}_l(\matr{Z}_l^\top\matr{Z}_l)^{-1/2}$ and $\matr{F}_l= (\matr{Z}_l^\top\matr{Z}_l)^{1/2}\matr{\Theta}_l(\matr{Z}_l^\top\matr{Z}_l)^{1/2}$, $l = 1,...,L$, as the unique
solution up to an orthogonal matrix, provided at least one of the $\matr{\Theta}_l$ is
full rank.
\label{lem:recovery}
\end{lemma}

\begin{proof}
To prove lemma \ref{lem:recovery}, we can show that $\matr{H'}_l= \matr{Z}_l(\matr{Z}_l^\top\matr{Z}_l)^{-1/2}$, $\matr{F}_l= (\matr{Z}_l^\top\matr{Z}_l)^{1/2}\matr{\Theta}_l(\matr{Z}_l^\top\matr{Z}_l)^{1/2}$ is a
solution to the optimization problem in \eqref{eq:ObjFnEq}. Substituting the solution to \eqref{eq:ObjFnEq}, we have

\footnotesize{
\begin{equation*}
\centering
    \begin{split}
        \sum_{l=1}^{L} || \matr{A}_{l}-\matr{Z}_l(\matr{Z}_l^\top\matr{Z}_l)^{-1/2}&(\matr{Z}_l^\top\matr{Z}_l)^{1/2}\matr{\Theta}_l(\matr{Z}_l^\top\matr{Z}_l)^{1/2}(\matr{Z}_l^\top\matr{Z}_l)^{-1/2}\matr{Z}_l^\top||^2_{F}\\
        &=\sum_{l=1}^{L} || \matr{A}_{l}-\matr{Z}_l\matr{\Theta}_l\matr{Z}_l^\top ||^2_{F}
    \end{split}
\end{equation*}}

\normalsize

\noindent and, since $\mathbb{E}(\matr{A}_l)= \matr{Z}_l\matr{\Theta}_l\matr{Z}_l^\top$, the value of this minimization objective function is 0, and $\matr{H'}_l^\top\matr{H'}_l= (\matr{Z}_l^\top\matr{Z}_l)^{-1/2}\matr{Z}_l^\top \matr{Z}_l(\matr{Z}_l^\top\matr{Z}_l)^{-1/2}=(\matr{Z}_l^\top\matr{Z}_l)^{1/2}(\matr{Z}_l^\top\matr{Z}_l)^{-1/2}=\matr{I}$

\begin{figure*}[ht]
            \centering \begin{subfigure}[b]{0.9\linewidth} \centering\includegraphics[width=0.9\linewidth]{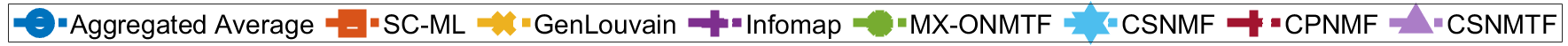}
            \end{subfigure}
            \\
            \centering \begin{subfigure}[b]{0.32\linewidth} \includegraphics[width=0.99\linewidth]{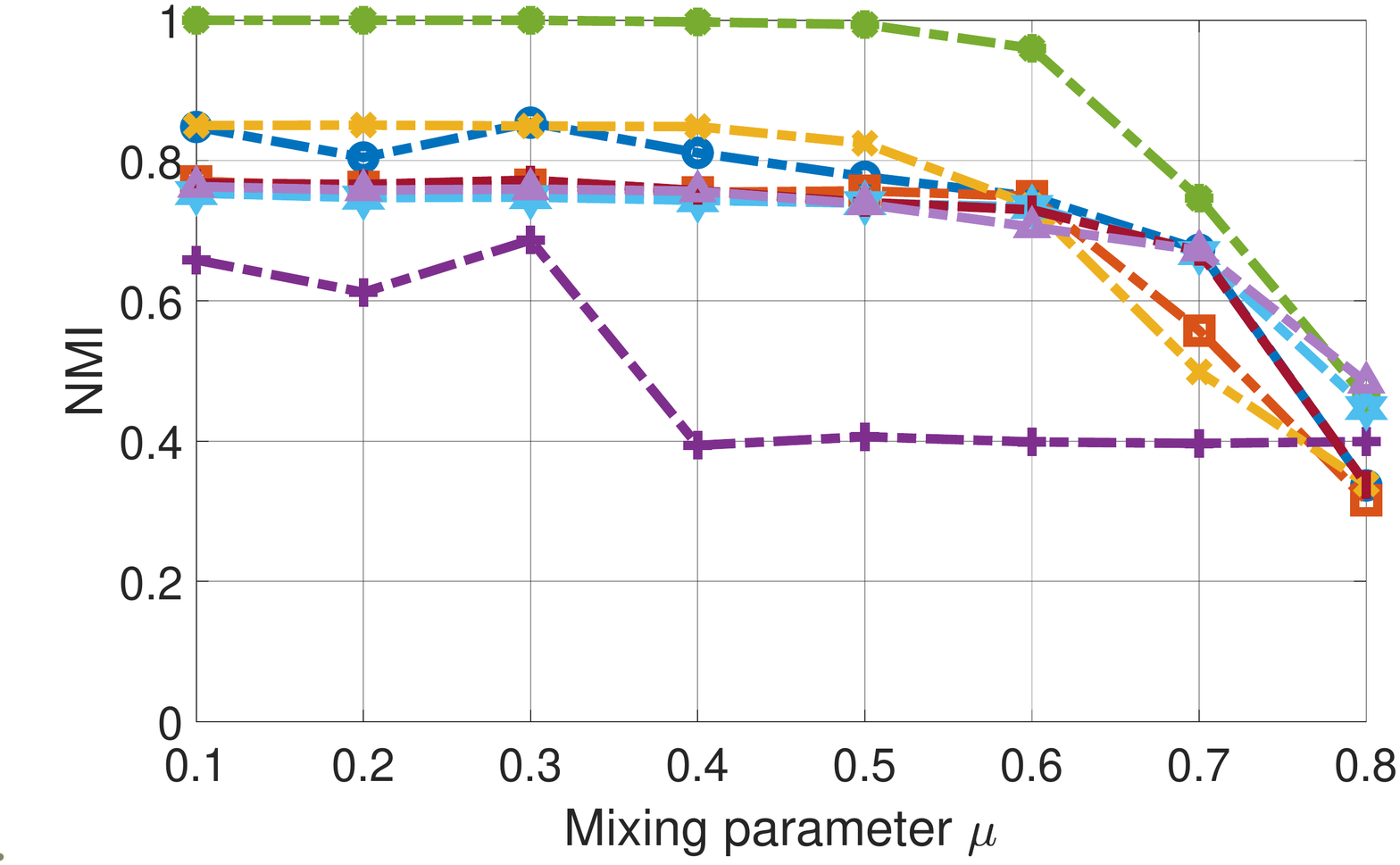}
            \caption{}
            \label{fig:2cc3L}
            \end{subfigure}
            \hfill
          \begin{subfigure}[b]{0.32\linewidth} \includegraphics[width=0.99\linewidth]{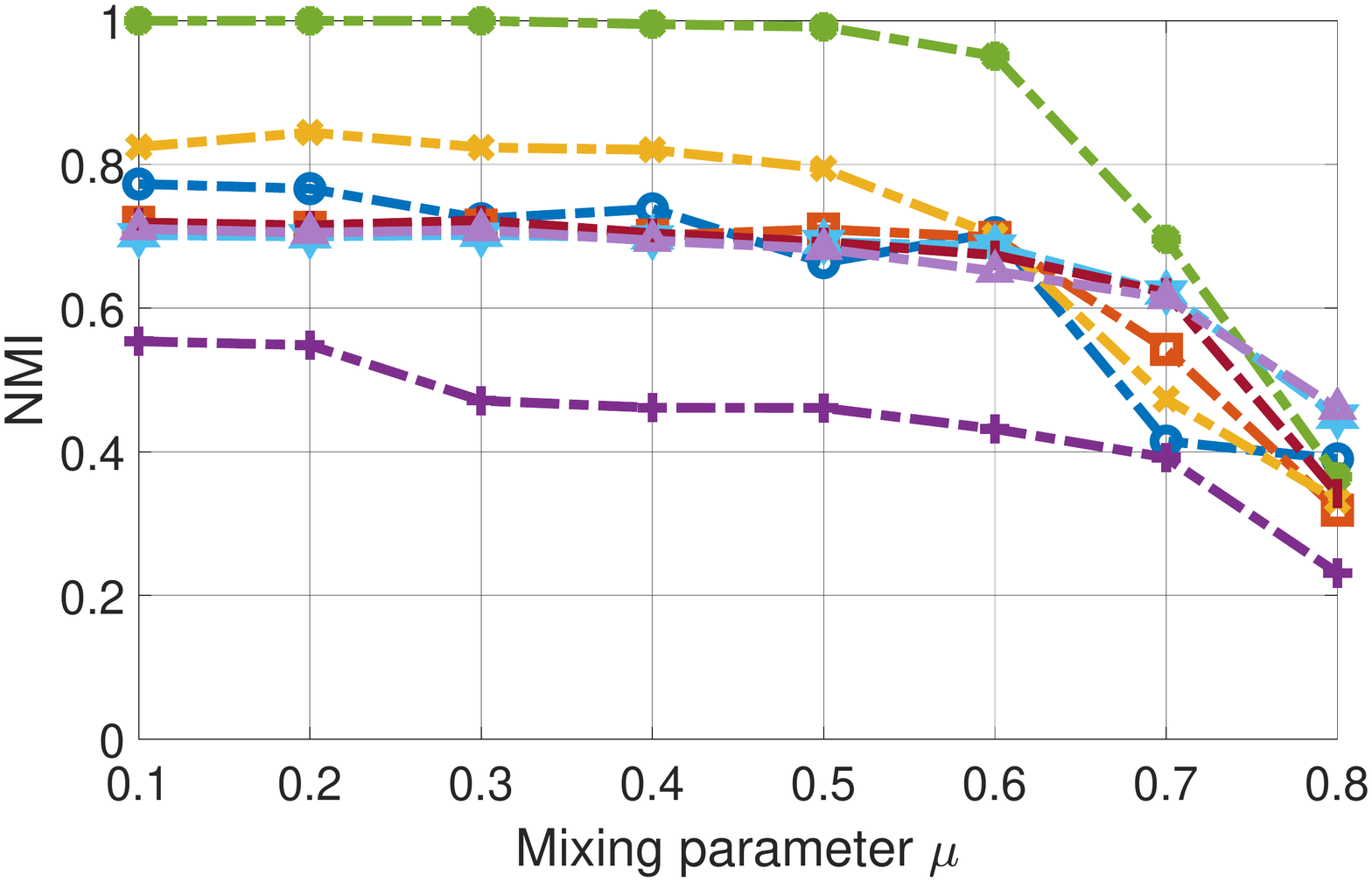}
            \caption{}
            \label{fig:2cc4L}
            \end{subfigure}
            \hfill
          \begin{subfigure}[b]{0.32\linewidth} \includegraphics[width=0.99\linewidth]{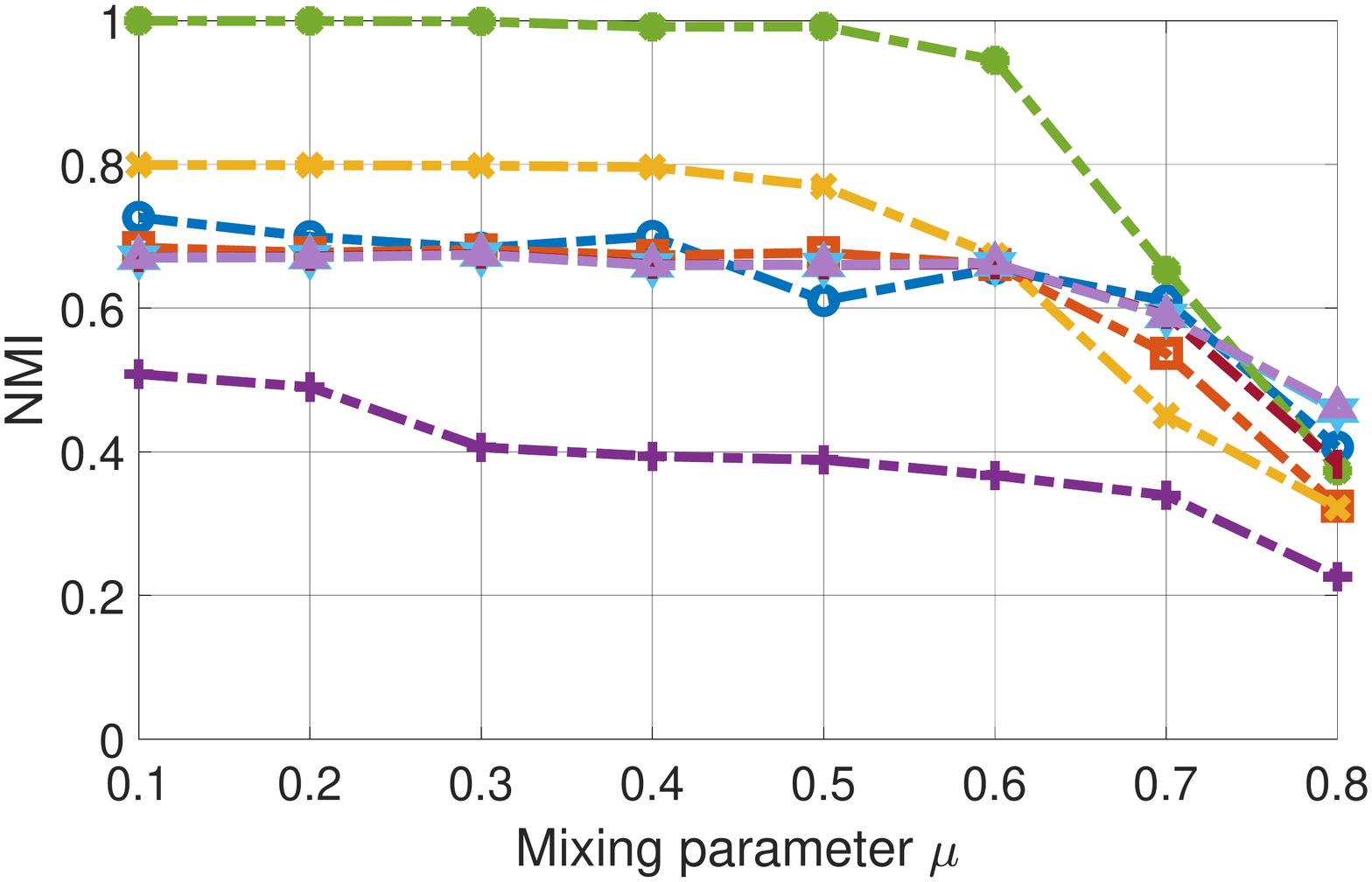}
            \caption{}
            \label{fig:2cc5L}
            \end{subfigure}
            \\
            \begin{subfigure}[b]{0.32\linewidth} \includegraphics[width=0.99\linewidth]{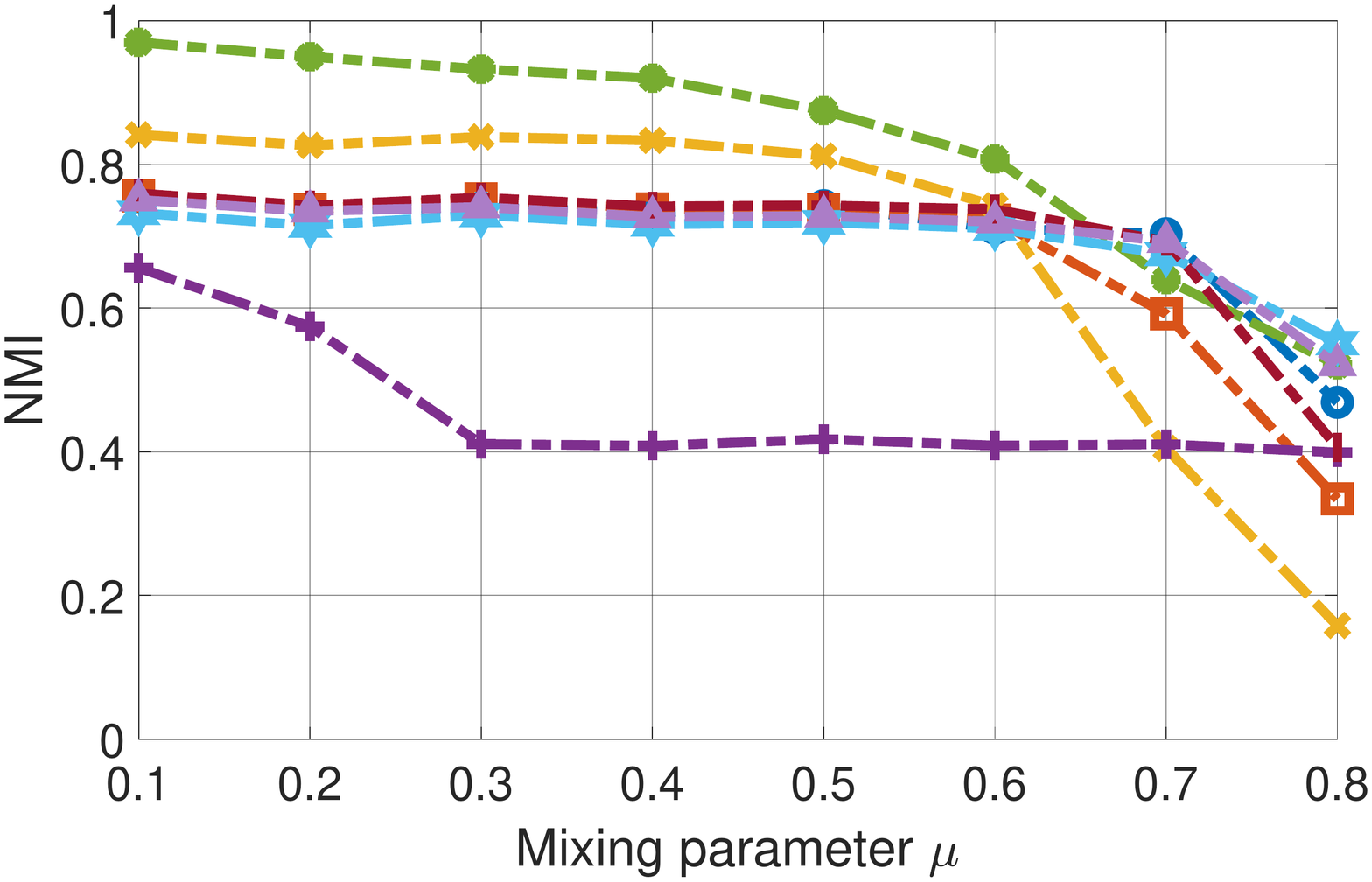}
            \caption{}
            \label{fig:3cc3L}
            \end{subfigure}
            \begin{subfigure}[b]{0.32\linewidth} \includegraphics[width=0.99\linewidth]{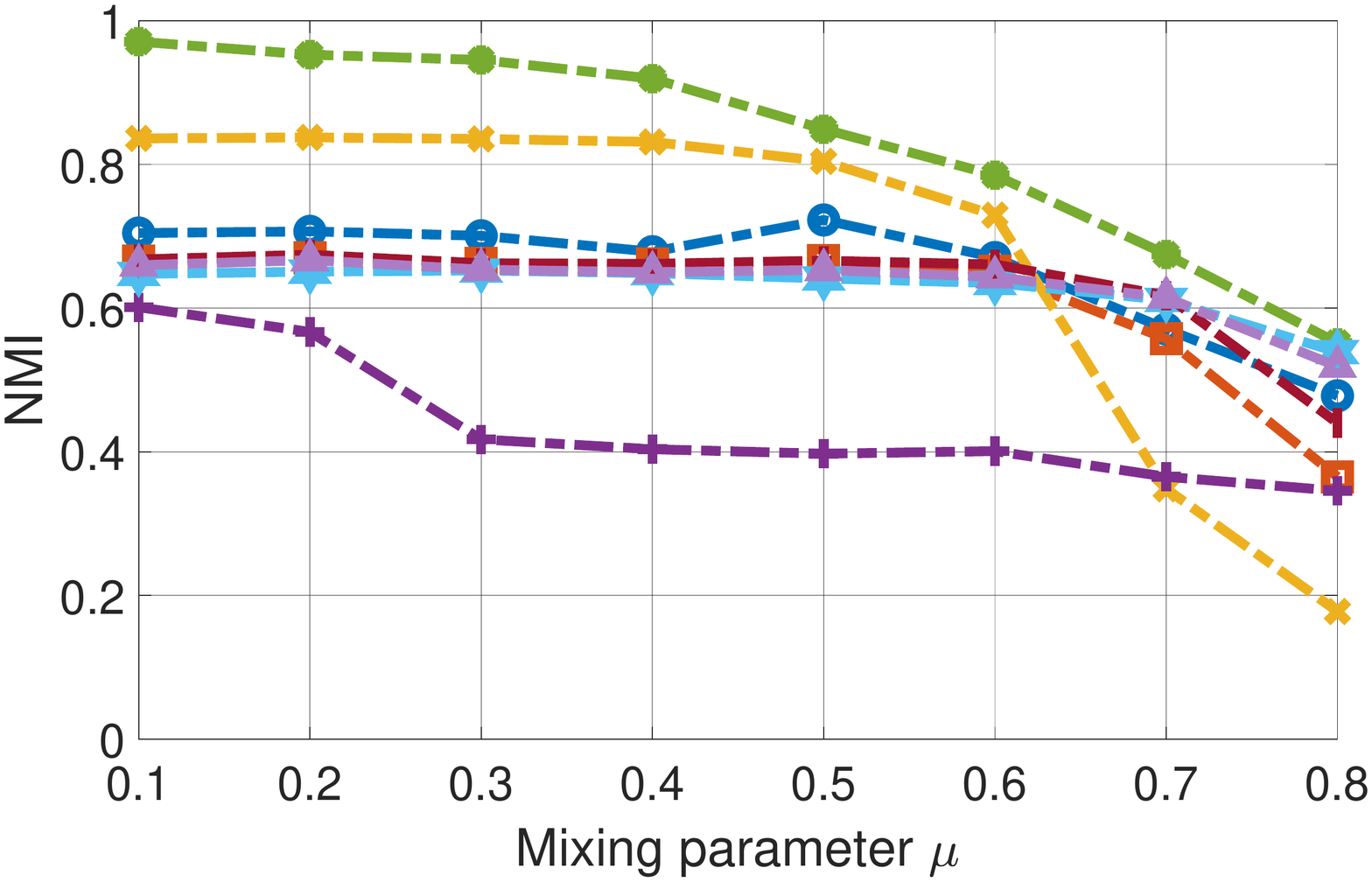}
            \caption{}
            \label{fig:3cc4L}
            \end{subfigure}
            \begin{subfigure}[b]{0.32\linewidth} \includegraphics[width=0.99\linewidth]{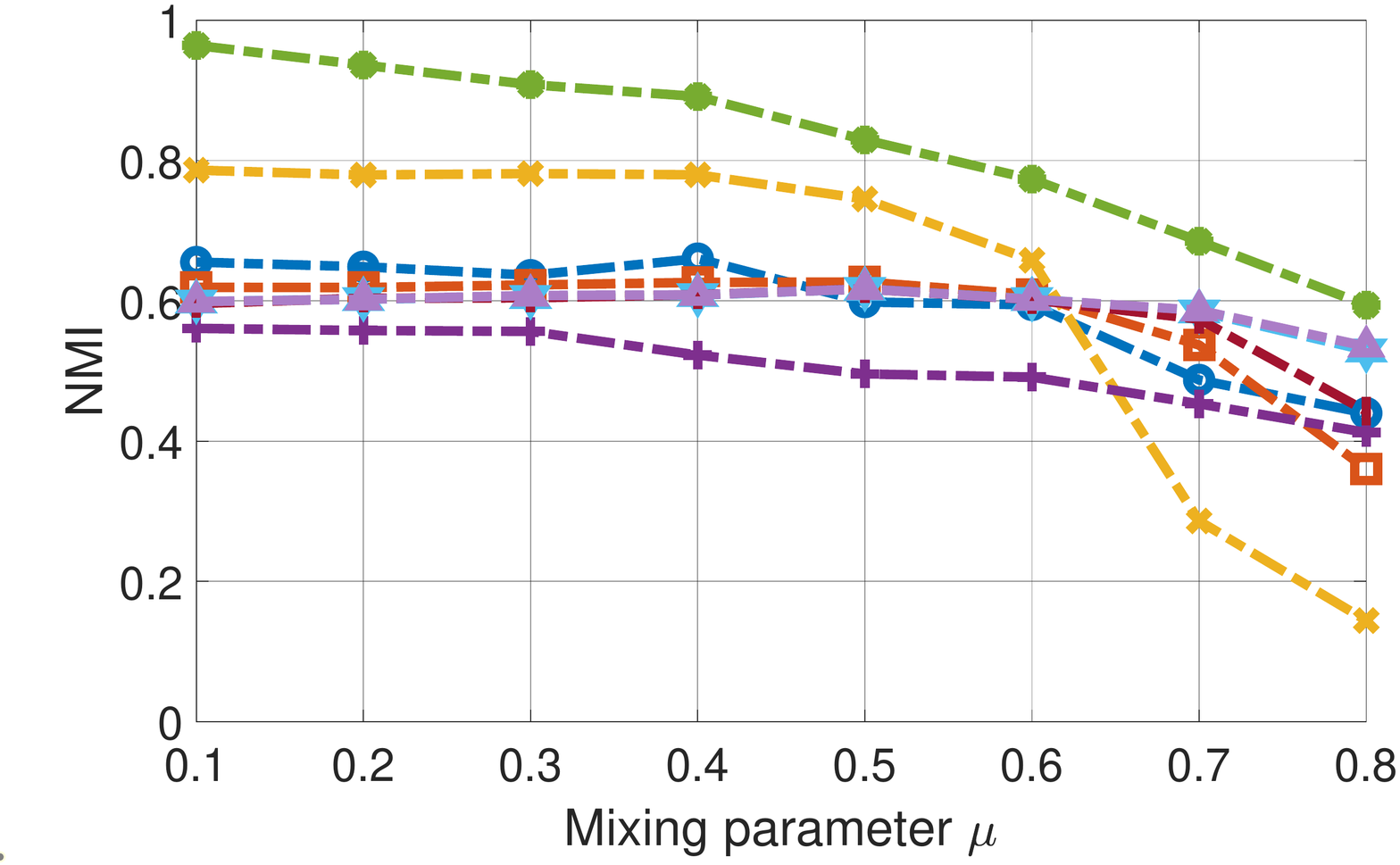}
            \caption{}
            \label{fig:3cc5L}
            \end{subfigure}
            \caption{Mean NMI over 100 realizations of (a)-(c) 3-layer, 4-layer and 5-layer benchmark networks, respectively, for the scenario with 2 common communities across all layers; (d)-(f) 3-layer, 4-layer and 5-layer benchmark networks, respectively, for the scenario with 3 common communities across different subsets of layers. All networks are generated with 8 different values of the mixing parameter $\mu$  and $n=256$.} 
            \label{fig:changemu}
\end{figure*}

Now, we need to show the uniqueness of this solution. By assumption, at
least one of the $\matr{\Theta}_l$ is full rank. For a non-singular matrix
$\matr{P}\in \mathbb{R}^{(k_c+k_{p_l})\times (k_c+k_{p_l})}$ we can say that $\matr{H'}_l\matr{P}$ and $\matr{P}^{-1}\matr{\Theta}_l\matr{P}^{{-1}^\top}$ is also a solution. 
Due to the orthogonality constraint, we must have $(\matr{H'}_l\matr{P})^\top\matr{H'}_l\matr{P}=\matr{P}^\top(\matr{Z}_l^\top\matr{Z}_l)^{-1/2}\matr{Z}_l^\top\matr{Z}_l(\matr{Z}_l^\top\matr{Z}_l)^{-1/2}\matr{P}=\matr{I}$, which implies $\matr{P}^\top\matr{P}=\matr{I}$, and therefore the solution is unique up to an orthogonal matrix. Moreover, since $\matr{Q}^{-1/2} = (\matr{Z}_l^\top\matr{Z}_l)^{-1/2}$ is a diagonal matrix with positive elements and therefore invertible, we have
that $\matr{Z}_i\matr{Q}^{-1/2} = \matr{Z}_j\matr{Q}^{-1/2}$ implies $\matr{Z}_i = \matr{Z}_j$.
\end{proof}

\section{Experiments}\label{sec:experiments}
\subsection{Synthetic Multiplex Networks}
\subsubsection{Model description}
Multiplex benchmark networks based on the model described in \cite{bazzi2016generative,jeub2016generative} were generated. 
The authors in \cite{bazzi2016generative} propose a two-step approach to generate multilayer networks with a community structure. First, a multilayer partition with the user-defined number of nodes in each layer, number of layers, and an interlayer dependency tensor that specifies the desired dependency structure between layers is generated. Next, for the given multilayer partition, edges in each layer are generated following a degree-corrected block model \cite{karrer2011stochastic} parameterized by
the distribution of expected degrees and a community mixing parameter $\mu \in [0,1]$. The mixing parameter $\mu$ controls the modularity of the network. When $\mu = 0$, all edges lie within communities, whereas $\mu=1$ implies that edges are distributed independently.
For multiplex networks, the probabilities in the interlayer dependency tensor are the same for all pairs of layers and are specified by $p\in [0,1]$. When $p=0$, the partitions are independent across layers while $p=1$ indicates an identical partition across layers.

In this paper, we extend the model described above to generate multiplex benchmark networks with common and private communities. We first generate the common communities by randomly selecting $n_c$ nodes  across all layers and setting the inter-layer dependency probability to $p_1$. For each common community, we decide whether it exists in a particular layer or not. Next, we independently generate the private communities for each layer with the remaining nodes in that layer.
 We generated 100 different random realizations of each multiplex network in order to
report the average performance metric on the experiments.

\subsubsection{Evaluation}
 We compared the performance of our method to well-known multiplex community detection algorithms. In particular, we compared with ONMTF applied to the aggregated multiplex networks using the average of the adjacency matrices (Aggregated Average), Spectral Clustering on Multi-Layer graphs (SC-ML) \cite{dong2013clustering}, Generalized Louvain (GenLouvain) multilayer community detection algorithm \cite{jutla2011generalized,mucha_community_2010}, Infomap \cite{de2015identifying},  Collective Symmetric Nonnegative Matrix Factorization (CSNMF) \cite{gligorijevic2018non}, Collective Projective Nonnegative Matrix Factorization (CPNMF) \cite{gligorijevic2018non}, and Collective Symmetric Nonnegative Matrix Tri-factorization (CSNMTF) \cite{gligorijevic2018non}. 

\subsubsection{Experiment 1}\label{exp:changemu}
In this experiment, we generated two different types of networks, one where the common communities are present across all layers and another where the common communities are present in different subsets of layers. 
Fig. \ref{fig:changemu} shows the results for the networks with 2 common communities across all layers for  3 (\ref{fig:2cc3L}), 4 (\ref{fig:2cc4L}), and 5 layers (\ref{fig:2cc5L}), and for the networks with 3 common communities across different subsets of layers for 3 (\ref{fig:3cc3L}), 4 (\ref{fig:3cc4L}), and 5 layers (\ref{fig:3cc5L}). The results indicate that our method performs well for both networks with common communities across all layers as well as for networks with common communities that do not span  all layers. Our method discovers the complete structure of the network rather than forcing it to have a consensus partition. Moreover, our method is robust to noise for larger values of $\mu$ compared to the other methods. We can also conclude that our algorithm performs better when the common communities are across all layers (see Fig. \ref{fig:2cc3L}, \ref{fig:2cc4L}, \ref{fig:2cc5L}) than when the multiplex community structure is more complex with common communities across subsets of layers (see Fig. \ref{fig:3cc3L}, \ref{fig:3cc4L}, \ref{fig:3cc5L} ), but it still outperforms the rest of the methods. From Fig. \ref{fig:changemu} 
 we can see that GenLouvain performs well when $\mu$ is small, but its performance deteriorates for $\mu$ values above 0.6. Another observation is that when the number of layers is small, the NMF-based methods perform closer to GenLouvain but when the number of layers increases, NMF algorithms perform worse. This is because these NMF methods perform aggregation on either the adjacency or the community indicator matrices. When there is more variation across layers, these methods fail to capture this heterogeneity.

\subsubsection{Experiment 2}\label{exp:changep1}
In the second experiment, we evaluated the robustness of the algorithm against variations in the common community structure  by fixing $\mu=0.1$ and varying the inter-layer dependency probability, $p_1$, i.e., the common communities are allowed to vary across layers. The performance of all methods for a 5-layer network are reported in Fig. \ref{fig:2cc5Lp} based on the average NMI over 100 realizations of the network. 
As we can see in Fig. \ref{fig:2cc5Lp}, our method still outperforms the other seven methods when there is some variation in the common community structure. This demonstrates that our method is robust to variations of the common community structure across layers. However, our algorithm is more sensitive to the drop in $p_1$ than the rest of the methods. This is because when the common communities have a high variation our algorithm may try to assign some of those nodes to the private communities. 

\begin{figure}[h!]
          \vspace{-0.5cm}
        \includegraphics[width=0.99\linewidth]{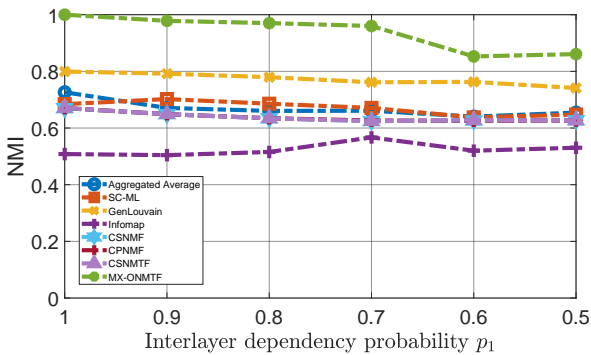}
         \vspace{-0.5cm}
     \caption{5-layer network generated with 6 different values of the interlayer dependency probability $p_1$, with  $\mu=0.1$, and $n=256$}
    \label{fig:2cc5Lp}
\end{figure}

\subsubsection{Experiment 3}\label{exp:changekc}
Another parameter in our model is the number of common communities $k_c$. In this experiment, we fixed $n=256$, $\mu=0.3$, and the number of communities in each layer, and varied $k_c$ from 1 to 7. When $k_c=7$, all communities are common across layers and there are no private communities. As we can see in Fig. \ref{fig:Expkcvar}, as $k_c$ is increased, the performance of all the other methods improves, as expected, because these methods are designed to detect the common community structure. When the communities are common across layers, most of the methods, except Infomap, converge to the same NMI value. The performance of MX-ONMTF is not affected by increasing $k_c$, and when all communities are common it performs similarly to the other methods.

\begin{figure}[h!]
          \vspace{-0.5cm} \includegraphics[width=0.99\linewidth]{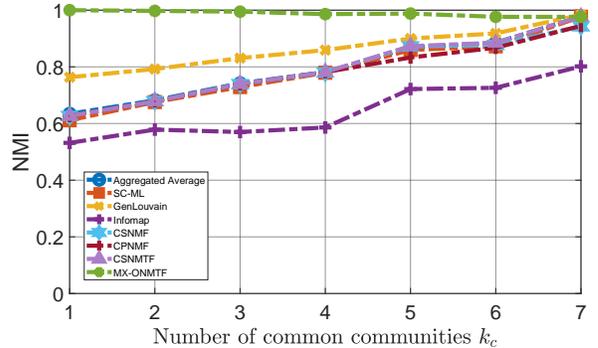}
                \vspace{-0.5cm}
     \caption{5-layer network generated with different values of common communities $k_c$ across layers.}
    \label{fig:Expkcvar}
\end{figure}

\subsubsection{Scalabilty Analysis}\label{exp:changeN}
In this experiment, we evaluate the effect of network size on the run time of the proposed algorithm. For this purpose, we fixed $\mu=0.3$, $L=5$, $k_c=3$ and varied $n$  from 32 to 8192. From Fig. \ref{fig:ExpNvar}, it can be seen that our method's run time is almost log-linear. This is comparable with all the other NMF based methods. However, our  as shown in the previous experiments, our method performs better. Most of this time complexity is due to the multiplicative update rule used in NMF-based algorithms and can be reduced using alternative approaches as discussed in \cite{vcopar2019fast}.

\begin{figure}[h!]
          \vspace{-0.5cm} \includegraphics[width=0.99\linewidth]{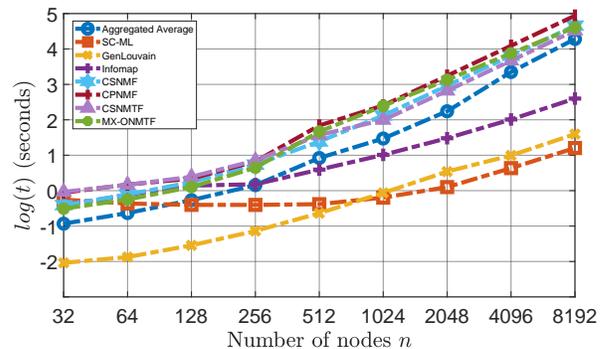}
     \caption{5-layer network with 9 different values of $n$, $k_c=3$, and $\mu=0.3$.}
    \label{fig:ExpNvar}
\end{figure}

\begin{figure*}[h]
            \centering \begin{subfigure}[b]{0.32\linewidth} \includegraphics[width=0.99\linewidth]{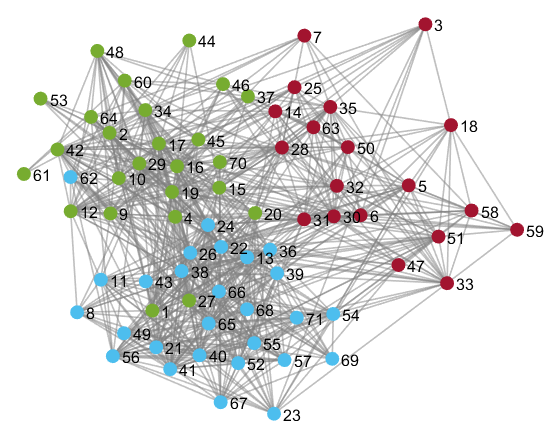}
            \caption{Advice}
            \end{subfigure}
            \hfill
          \begin{subfigure}[b]{0.32\linewidth} \includegraphics[width=0.99\linewidth]{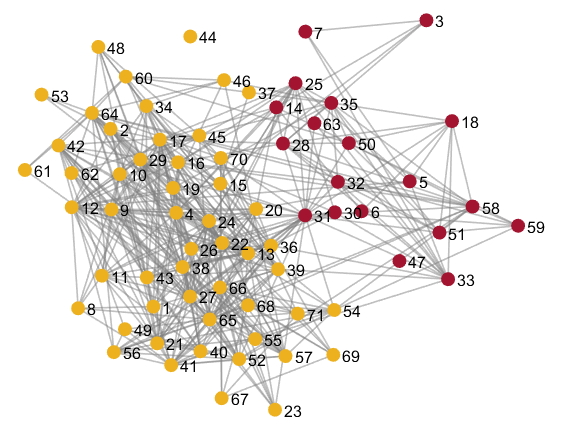}
            \caption{Friendship}
            \end{subfigure}
            \hfill
          \begin{subfigure}[b]{0.32\linewidth} \includegraphics[width=0.99\linewidth]{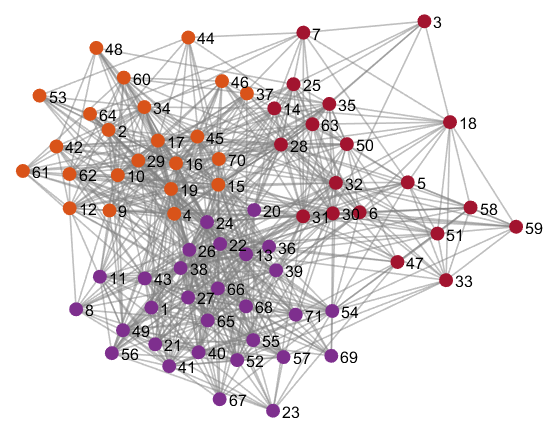}
            \caption{Co-work}
            \end{subfigure}
            \caption{Communities detected in Lazega Law Firm network across the three layers, advice, friendship, and co-work relationships. Red nodes  are in the common community across the three layers.}
            \label{fig:Laz3L}
\end{figure*}
\begin{table*}[h]
    \centering
    \begin{tabular}{cccccccc}
    \hline
     Method & Status & Gender & Office & Seniority & Age & Practice & Law School\\
    \hline
    GenLouvain & 0.0345 & 0.0307 & 0.5294 & 0.0807 & 0.0431 & 0.5468 & 0.0040 \\
Aggregated Average & 0.0383 & 0.0197 & 0.5379 & 0.1307 & 0.0798 & 0.4411 & 0.0201 \\
    SC\_ML & 0.0138 & 0.0259 & 0.0731 & 0.1225 & 0.0464 & 0.0249 & 0.0140 \\
    Infomap &  0.0179 & 0.0043 & 0.1668 & 0.2880& 0.0083 & 0.0003 & 0.0093 \\
      CSNMF &	0.0418 & 0.0291 & 0.5732 &	0.1155	& 0.0736 & 0.4227 &	0.0172\\
CPNMF &	0.0081 & 0.0524 & 0.1139 & 0.0514 & 0.0291 & 0.0187 & 0.0221\\
CSNMTF & 0.0395 & 0.0217 & 0.0798 &	0.0795 &	0.0487 & 0.1335 & 0.0279\\
    MX-ONMTF & \textbf{0.4752} & \textbf{0.4906} & \textbf{0.7386} & \textbf{0.4135} & \textbf{0.4203} & \textbf{0.6162} & \textbf{0.4226} \\
     \hline
    \end{tabular}
    \caption{NMI of the obtained community partition for each method and the metadata available for the Lazega Law Firm Multiplex Network.}
    \label{tab:NMIlaz}
\end{table*}
\subsection{Real World Multiplex Networks}
\subsubsection{Lazega Law Firm Multiplex Social Network}
Lazega Law Firm \cite{lazega2001collegial} is a multiplex social network with 71 nodes and three layers representing Co-work, Friendship and Advice relationships between partners and associates of a corporate law firm. This data set also includes information about some attributes of each node such as status, gender, office location, years with the firm, age, type of practice, and law school.

Applying MX-ONMTF to this network, we obtain one common community across all layers composed of the nodes colored in red as well as private communities for each layer, as shown in Fig. \ref{fig:Laz3L}. This network does not have ground truth community structure, but we can compute the NMI between the detected community structure and each type of node attributes, i.e., metadata, to gain better insight into the results and to be able to provide quantitative results \cite{roxana_pamfil_relating_2018}. For each of the attributes, the nodes are divided into communities based on that particular attribute. For example, for the status, the network is divided into two communities, partners and associates. For Age and Seniority, the nodes were grouped into five-year bins. The community structure for each attribute is used as ground truth to compute the NMI between each attribute and the community structure detected by our method. The NMI values given in Table \ref{tab:NMIlaz} for the partition obtained by our method, suggest that office location and type of practice (litigation or corporate) are highly correlated with community membership  across co-work, friendship and advice relationships. We can also see that the partition detected by MX-ONMTF has greater NMI values for each of the attributes. Therefore, our method detects a community structure that takes all of the attributes into account instead of partitioning with respect to just one attribute as the Aggregated Average does.

\subsubsection{C. Elegans Network}
C. Elegans Network \cite{de2015muxviz,chen2006wiring} is a multiplex network with 279 nodes and 3 layers representing different synaptic junctions (electric, chemical monadic, and polyadic) of 279 neurons of the Caenorhabditis Elegans connectome. Information about different attributes of the neurons in this dataset such as the group of neuron they belong to (bodywall, mechanosensory, ring interneurons, head motor neurons, etc.), the type of neuron (motor neurons, sensory neurons, interneurons), and the color (blue, red, yellow, orange, etc.) is available.

Table \ref{tab:NMIce} shows the NMI values between the community structures detected by each method and each of the three attributes available for this dataset.  The partition detected by MX-ONMTF has greater NMI values for each of the attributes compared to the other four methods.

\begin{table}[H]
  \centering
    \begin{tabular}{cccc}
    \hline
  Method & Neuron Group & Neuron Type & Color \\
\hline
 GenLouvain & 0.3756 &	0.1297 & 0.2362 \\
 Aggregated Average & 0.3839 & 0.1590 & 0.2977 \\
    SC-ML &  0.0185 &	0.0103 &	0.2690 \\
    Infomap &  0.2265 & 0.2345 & 0.2355\\
      CSNMF & 0.1635 &	0.075 &	0.1211 \\
CPNMF & 0.0854 & 0.0277 & 0.0628\\
CSNMTF &  0.1113 & 0.0402 & 0.0914 \\
    MX-ONMTF & \textbf{0.4074} & \textbf{0.4001} & \textbf{0.4593} \\
    \hline
    \end{tabular}%
\caption{NMI of the obtained community partition for each method and the metadata available for C. Elegans Network.}
  \label{tab:NMIce}%
\end{table}%

\subsubsection{YeastLandscape Multiplex Network}
Yeast Landscape is a multiplex genetic interaction
network of a specie of yeast, Saccharomyces Cerevisiae \cite{costanzo2010genetic,de2015muxviz}.  This network has 4458 nodes and 4 layers representing the  positive and negative interaction networks of genes in Saccharomyces cerevisiae and positive and negative correlation based networks in which genes with similar interaction profiles are connected to each other. For this paper, we use the bioprocess annotations of the genes available on the supplementary data file S6 of \cite{costanzo2010genetic} as ground truth. We divided the genes into 18 groups according to their primary bioprocess. There were 1580 genes in this network without attributes. 

Table \ref{tab:NMIyl} shows the NMI values between the community structures detected by each method and the bioprocess of the genes. MX-ONMTF gives the highest NMI value followed by the other NMF-based community detection methods.
\begin{table}[h]
  \centering
    \begin{tabular}{cc}
    \hline
   Method & Bioprocess \\
\hline
 GenLouvain & 0.0794 \\
 Aggregated Average & 0.1108 \\
    SC-ML & 0.1564  \\
    Infomap &  0.2987\\
      CSNMF & 0.3553\\
CPNMF &	0.3559 \\
CSNMTF & 0.3549  \\
    MX-ONMTF & \textbf{0.4123} \\
    \hline
    \end{tabular}%
\caption{NMI of the obtained community partition for each method with respect to the metadata available for YeastLandscape Network.}
   \label{tab:NMIyl}%
\end{table}%

\subsection{Multiview Networks}
In order to evaluate the performance of our method on networks where the communities are common across all layers, we use two multiview data sets, UCI Handwritten Digits\footnote{https://archive.ics.uci.edu/ml/datasets/Multiple+Features}  \cite{Dua:2019} and Caltech \cite{caltech}. 

The UCI Handwritten Digits data set consists of features of handwritten
digits from (0- 9) extracted from a collection of Dutch utility maps. There is a total of 2000 patterns that have been digitized in binary images, 200 patterns per digit. These digits are represented by six different feature sets: Fourier coefficients of the character shapes, profile correlations, Karhunen-Lo\`{e}ve coefficients, pixel averages in 2 $\times$ 3 windows, Zernike moments, and morphological features. Each layer of the multiplex network represents one of the 6 features. The graphs are constructed using $k$-nearest neighbors graphs with the nearest 50 neighbors and Euclidean distance.

Caltech-101 is a well-known object recognition dataset that consists of pictures of objects belonging to 102 categories. There are about 40 to 800 images per category for a total of 9144 images. This dataset consists of 6 types of features extracted from each image. A multiplex network with 6 layers representing each of the features, 102 classes, and 9144 nodes is constructed from this dataset using k-nearest neighbors graphs with the nearest 50 neighbors. A smaller version of this dataset is also used in these experiments, where only 20 objects are selected, resulting in a multiplex network with 6 layers, 20 communities, and 2386 nodes.

In this case, as we have the true class assignment, we compute the NMI with respect to this ground truth. As it can be seen in Table \ref{tab:NMImv}, our method performs better than the rest of the methods for the three networks. This indicates that even in cases where there are no private communities, our method is successful at obtaining the consensus community structure, thus can be used as an alternative to multiview clustering.

\begin{table}[h]
  \centering
    \begin{tabular}{cccc}
    \hline
   Method & Handwritten & Caltech-20 & Caltech-101 \\
\hline
 GenLouvain & 0.8791 &	0.5921 & 0.3406 \\
 Aggregated Average & 0.7957 & 0.5358 & 0.3941  \\
    SC-ML &  0.8435 & 0.6476 & 0.5016\\
    Infomap & 0.5367 & 0.3876 & 0.2583 \\
      CSNMF & 0.4499	 & 0.4250 & 0.3842 \\
CPNMF & 0.4421	 & 0.4208 & 0.3816 \\
CSNMTF & 0.4478  & 0.4264 & 0.3831  \\
    MX-ONMTF & \textbf{0.9432} & \textbf{0.6861} & \textbf{0.5660} \\
    \hline
    \end{tabular}%
\caption{NMI of the obtained community partition for multiview networks.}
   \label{tab:NMImv}%
\end{table}%

\section{Conclusions}\label{sec:conclusions}
In this paper, we proposed a multiplex community detection method based on ONMTF. The proposed method, MX-ONMTF, is able to  detect both common and private communities across layers, allowing us to differentiate between the topologies across layers. The proposed algorithm is based on multiplicative update rules and a proof of  convergence 
is provided. A new approach based on the eigengap criterion is introduced for determining the number of communities. 
Results for  both synthetic and real-world networks show that our method performs better than existing community detection methods for multiplex networks as it is able to handle the heterogeneity of the network topology across layers.  Moreover, experiments on multiview networks show that our method also performs well in cases where a consensus community structure is needed. 

\appendices
\section{Auxiliary Function Proof}
\begin{proposition}
 The following function, $Z(h,h_{ij}^t)$, \begin{align*}
    \begin{split}
    &Z(h,h_{ij}^t)=\mathcal{L}(h^t_{ij})+3\mathcal{L'}(h^t_{ij})(h-h^t_{ij})+ \\
    &\frac{3}{2}\frac{\sum_{l=1}^{L}(4\matr{H}_l\matr{G}_l\matr{H}_l^\top\matr{H}^t\matr{S}_l+4\matr{H}^t{\matr{H}^t}^\top\matr{A}_l\matr{H}^t\matr{S}_l)_{ij}}{h_{ij}^t}(h-h^t_{ij})^2
    \end{split}
\end{align*}

is an auxiliary function of $\mathcal{L}(H)$,
\begin{align*}
\begin{split}
\mathcal{L}(h)=\mathcal{L}(h^t_{ij})+\mathcal{L'}(h^t_{ij})(h-h^t_{ij})+ \frac{1}{2}\mathcal{L''}(h_{ij}^t)(h-h^t_{ij})^2.
\end{split}
\end{align*}
\end{proposition}

\begin{proof} 
First,  when $h=h_{ij}^t$ the equality $ Z(h,h)=\mathcal{L}(h)$ holds. Now, we need to show that $Z(h,h_{ij}^t)\ge\mathcal{L}(h)$.

It can be seen that the first and second terms of $Z(h,h_{ij}^t)$ are greater than the first and second terms in $\mathcal{L}(h)$. Therefore, it suffices to show that $\frac{3\sum_{l=1}^{L}(4\matr{H}_l\matr{G}_l{\matr{H}_l}^\top\matr{H}^t\matr{S}_l+4\matr{H}^t{\matr{H}^t}^\top\matr{A}_l\matr{H}^t\matr{S}_l)_{ij}}{h_{ij}^t}\ge  \mathcal{L''}(h_{ij}^t)$.

It can be shown that

\begin{equation*}
\begin{split}
\frac{(\matr{H}_l\matr{G}_l\matr{H}_l^\top\matr{H}^t\matr{S}_l)_{ij}}{h_{ij}^t}=\frac{\sum_{p,q}(\matr{H}_l\matr{G}_l\matr{H}_l^\top)_{ip}h_{pq}^{qj}}{h_{ij}^t}\\
 \geq (\matr{H}_l\matr{G}_l\matr{H}_l^\top)_{ii}\matr{S}_{jj},
\end{split}
\end{equation*}
\vspace{-0.3cm}

\begin{equation*}
\begin{split}
\frac{(\matr{H}^t{\matr{H}}^{t^\top}\matr{A}_l\matr{H}^t\matr{S}_l)_{ij}}{h_{ij}^t}=\frac{\sum_{p}^{}h_{ip}^t({\matr{H}^t}^\top\matr{A}_l\matr{H}^t\matr{S}_l)_{pj}}{h_{ij}^t}\\
\geq {({\matr{H}^t}^\top\matr{A}_l\matr{H}^t\matr{S}_l)_{jj}}, 
\end{split}
\end{equation*}

\vspace{-0.3cm}

\begin{align*}
\frac{(\matr{H}^t{\matr{H}}^{t^\top}\matr{A}_l\matr{H}^t\matr{S}_l)_{ij}}{h_{ij}^t}=\frac{\sum_{p,q}^{}h_{ip}^th_{qp}^t(\matr{A}_l\matr{H}^t\matr{S}_l)_{qj}}{h_{ij}^t}\\
\ge h_{ij}^t(\matr{A}_l\matr{H}^t\matr{S}_l)_{ij},
\end{align*}

\vspace{-0.3cm}
\begin{align*}
\frac{(\matr{H}^t{\matr{H}}^{t^\top}\matr{A}_l\matr{H}^t\matr{S}_l)_{ij}}{h_{ij}^t}=\frac{\sum_{m,b}^{}(\matr{H}^t{\matr{H}}^{t^\top}\matr{A}_l)_{im}h_{mb}^{bj}}{h_{ij}^t}\\
\ge (\matr{H}^t{\matr{H}}^{t^\top}\matr{A}_l)_{ii}\matr{S}_{l_{jj}}.\\
\end{align*}

Therefore,\\
\begin{align*}
    \begin{split}
3\frac{(\matr{H}^t{\matr{H}}^{t^\top}\matr{A}_l\matr{H}^t\matr{S}_l)_{ij}}{h_{ij}^t}&\ge ({\matr{H}^t}^\top\matr{A}_l\matr{H}^t\matr{S}_l)_{ij}+h_{ij}^t(\matr{A}_l\matr{H}^t\matr{S}_l)_{ij}\\+&(\matr{H}^t{\matr{H}}^{t^\top}\matr{A}_l)_{ii}\matr{S}_{l_{jj}},
    \end{split}
\end{align*}

\noindent and thus $\frac{3\sum_{l=1}^{L}(4\matr{H}_l\matr{G}_l\matr{H}_l^\top\matr{H}^t\matr{S}_l+4\matr{H}^t\matr{H}^{t^\top}\matr{A}_l\matr{H}^t\matr{S}_l)_{ij}}{h_{ij}^t}\ge  \mathcal{L''}(h_{ij}^t)$. Therefore, Eq. \eqref{eq:auxiliary} is an auxiliary function of $\mathcal{L}(h)$.
\end{proof}

\ifCLASSOPTIONcompsoc
  \section*{Acknowledgments}
\else
  \section*{Acknowledgment}
\fi

This work was supported in part by National Science Foundation under Grants CCF-2006800.

\ifCLASSOPTIONcaptionsoff
  \newpage
\fi

\bibliographystyle{IEEEtran}
\bibliography{main}
\vspace{-1cm}

\end{document}